\documentclass[a4paper,onecolumn,superscriptaddress,11pt,%
        allowfontchangeintitle,%
        accepted=2025-07-21
				]{quantumarticle}
\pdfoutput=1

\usepackage[utf8]{inputenc}
\usepackage[british]{babel}
\usepackage{amsthm,amsmath}
\usepackage[pdfusetitle]{hyperref}
\usepackage[numbers,sort&compress]{natbib}
\usepackage[all]{hypcap}
\usepackage{stmaryrd}
\usepackage{graphicx}
\usepackage{keycommand}
\usepackage{multicol}
\usepackage{array}

\theoremstyle{definition}
\newtheorem{theorem}{Theorem}[section]
\newtheorem{corollary}[theorem]{Corollary}
\newtheorem{lemma}[theorem]{Lemma}
\newtheorem{proposition}[theorem]{Proposition}

\newtheorem{definition}[theorem]{Definition}

\newtheorem{example*}[theorem]{Example*}
\newtheorem{examples*}[theorem]{Examples*}
\newtheorem{remark}[theorem]{Remark}
\newtheorem{remark*}[theorem]{Remark*}

\newtheorem*{theorem*}{Theorem}
\newtheorem*{corollary*}{Corollary}
\newtheorem*{lemma*}{Lemma}
\newtheorem*{proposition*}{Proposition}

\def\<{\langle}
\def\>{\rangle}
\usepackage{tikzit}
\input{zh.tikzdefs}

\tikzstyle{dot}=[inner sep=0.3mm, minimum width=2mm, minimum height=2mm, draw, shape=circle, font={\footnotesize}, tikzit fill=magenta]
\tikzstyle{hadamard}=[fill=zxhad, draw, inner sep=0.6mm, minimum height=1.5mm, minimum width=1.5mm, shape=rectangle, tikzit shape=rectangle, tikzit category=ZH-pf, tikzit fill=yellow]
\tikzstyle{small hadamard}=[hadamard]
\tikzstyle{lambda}=[hadamard, fill={rgb,255: red,180; green,180; blue,180}, tikzit shape=rectangle]
\tikzstyle{halfscalar}=[star, fill=black, draw=black, minimum size=8pt, inner sep=0pt]
\tikzstyle{box}=[shape=rectangle, text height=1.5ex, text depth=0.25ex, yshift=0.2mm, fill=white, draw=black, minimum height=3mm, minimum width=5mm, font={\small}]
\tikzstyle{Z dot}=[inner sep=0mm, minimum size=2mm, shape=circle, draw=black, fill={zx_green}, tikzit fill=green]
\tikzstyle{Z phase dot}=[minimum size=5mm, font={\footnotesize\boldmath}, shape=rectangle, rounded corners=2mm, inner sep=0.2mm, outer sep=-2mm, scale=0.8, tikzit shape=circle, draw=black, fill={zx_green}, tikzit draw=blue, tikzit fill=green]
\tikzstyle{X dot}=[Z dot, shape=circle, draw=black, fill={zx_red}, tikzit fill=red]
\tikzstyle{X phase dot}=[Z phase dot, tikzit shape=circle, tikzit draw=blue, fill={zx_red}, font={\footnotesize\color{black}\boldmath}, tikzit fill=red]
\tikzstyle{H box}=[hadamard]
\tikzstyle{st}=[star, star points=5, fill=white, draw=black, inner sep=1.2pt, line width=1.2pt, tikzit fill=blue, tikzit draw=red, tikzit category=ZH-pf]
\tikzstyle{triangle}=[regular polygon, regular polygon sides=3, fill=white, draw=black, inner sep=0pt, minimum width=1em, tikzit draw=blue, tikzit category=ZH-pf, tikzit fill=cyan]
\tikzstyle{not}=[fill={rgb,255: red,180; green,180; blue,180}, draw=black, shape=circle, font={$\neg$}, dot]
\tikzstyle{vertex}=[inner sep=0mm, minimum size=1mm, shape=circle, draw=black, fill=black]
\tikzstyle{vertex set}=[inner sep=0mm, minimum size=1mm, shape=circle, draw=black, fill=white, font={\footnotesize\boldmath}]
\tikzstyle{wide point}=[fill=white, draw, shape=isosceles triangle, shape border rotate=-90, isosceles triangle stretches=true, inner sep=0pt, minimum width=1.5cm, minimum height=6.12mm, yshift=-0.0mm]
\tikzstyle{medium gray box}=[semilarge box, fill={rgb,255: red,180; green,180; blue,180}]
\tikzstyle{small box}=[rectangle, fill=white, draw, minimum height=5mm, yshift=-0.5mm, minimum width=5mm, font={\small}]
\tikzstyle{small gray box}=[small box, fill={rgb,255: red,180; green,180; blue,180}]
\tikzstyle{medium box}=[rectangle, fill=white, draw, minimum height=10mm, yshift=-0.5mm, minimum width=8mm, font={\small}, tikzit shape=rectangle]
\tikzstyle{ddot}=[line width=1.6pt, inner sep=0mm, minimum width=2.5mm, minimum height=2.5mm, draw, shape=circle]
\tikzstyle{dd white}=[ddot, fill=white, tikzit draw=green]
\tikzstyle{dd white phase}=[white phase dot, line width=1.6pt, tikzit draw=yellow]
\tikzstyle{dd gray}=[ddot, fill={rgb,255: red,180; green,180; blue,180}, tikzit draw=green]
\tikzstyle{dd gray phase}=[gray phase dot, line width=1.6pt, tikzit draw=yellow]
\tikzstyle{empty diagram}=[draw={gray!40!white}, dashed, shape=rectangle, minimum width=1cm, minimum height=1cm]
\tikzstyle{empty diagram small}=[draw={gray!50!white}, dashed, shape=rectangle, minimum width=0.6cm, minimum height=0.5cm]
\tikzstyle{white dot}=[Z dot, tikzit fill=green]
\tikzstyle{white phase dot}=[Z phase dot, tikzit shape=circle, tikzit fill=green, tikzit draw=blue]
\tikzstyle{gray dot}=[X dot, tikzit fill=red]
\tikzstyle{gray phase dot}=[X phase dot, tikzit shape=circle, tikzit draw=blue, tikzit fill=red]

\tikzstyle{simple}=[-]
\tikzstyle{hadamard edge}=[-, dashed, dash pattern=on 2pt off 1pt, thick, draw=blue]
\tikzstyle{gray}=[-, draw={blue!60!white}, tikzit draw=blue]
\tikzstyle{blue}=[-, draw={blue!60!white}, tikzit draw=blue]
\tikzstyle{brace edge}=[-, tikzit draw=blue, decorate, decoration={brace,amplitude=1mm,raise=-1mm}]
\tikzstyle{diredge}=[->]
\tikzstyle{not edge}=[-, dashed, dash pattern=on 2pt off 1.5pt, thick, draw={rgb,255: red,255; green,68; blue,68}]
\tikzstyle{double edge}=[-, double, shorten <=-1mm, shorten >=-1mm, double distance=2pt]
\tikzstyle{boldedge}=[-, line width=1.6pt, shorten <=-0.17mm, shorten >=-0.17mm, tikzit draw=blue]

\usepackage[all]{hypcap} % Make hyperlinks go to top of figures

\makeatletter
\newcommand\etc{etc\@ifnextchar.{}{.\@}\xspace}
\makeatother

\usepackage{bm}

\newcommand{\norm}[1]{\ensuremath{\lVert#1\rVert}}
\newcommand{\CC}{\mathbb{C}}
\newcommand{\R}{\mathbb{R}}

\newcommand{\Z}{\mathbb{Z}}
\newcommand{\B}{\mathbb{F}_2}

\begin{document}
\title{\texorpdfstring{\LARGE Optimal compilation of parametrised quantum circuits}{Optimal compilation of parametrised quantum circuits}}
\date{July 21st 2025}

\author{John van de Wetering}
\email{john@vdwetering.name}
\homepage{http://vdwetering.name}
\affiliation{University of Amsterdam}

\author{Richie Yeung}
\email{richie.yeung@cs.ox.ac.uk}
\affiliation{University of Oxford}
\affiliation{Quantinuum}

\author{Tuomas Laakkonen}
\email{tuomas.laakkonen@quantinuum.com}
\affiliation{Quantinuum}

\author{Aleks Kissinger}
\email{aleks.kissinger@cs.ox.ac.uk}
\affiliation{University of Oxford}

\begin{abstract}
  Parametrised quantum circuits contain phase gates whose phase is determined by a classical algorithm prior to running the circuit on a quantum device. Such circuits are used in variational algorithms like QAOA and VQE. In order for these algorithms to be as efficient as possible it is important that we use the fewest number of parameters. We show that, while the general problem of minimising the number of parameters is NP-hard, when we restrict to circuits that are Clifford apart from parametrised phase gates and where each parameter is used just once, we \emph{can} efficiently find the optimal parameter count. We show that when parameter transformations are required to be sufficiently well-behaved, the only rewrites that reduce parameters correspond to simple `fusions'. Using this we find that a previous circuit optimisation strategy by some of the authors [Kissinger, van de Wetering. PRA (2019)] finds the optimal number of parameters.
  Our proof uses the ZX-calculus. We also prove that the standard rewrite rules of the ZX-calculus suffice to prove any equality between parametrised Clifford circuits.
\end{abstract} 

\maketitle

\section{Introduction}

A quantum circuit is built out of small unitary gates that together make it possible to perform an arbitrary quantum computation. In a \emph{parametrised} quantum circuit, we allow certain quantum gates to be specified by a classical parameter that is determined before running the circuit on a quantum device. Usually parametrised gates are either phase gates or controlled-phase gates, and the parameter, a real number, specifies the phase to be applied.
Parametrised quantum circuits are an increasingly important construction for quantum algorithms, especially for near-term applications. For instance, variational algorithms such as QAOA \cite{farhi2014quantum} and VQE \cite{peruzzo2014variational} use a feedback loop between a classical side and a quantum side where the parameters are updated by a classical optimisation procedure, based on measurement outcomes of the quantum device.
Each training step of a typical optimisation procedure involves estimating the gradient of the cost function with respect to each parameter (e.g. using the parameter-shift rule \cite{schuld2019evaluating}).
In order to be as efficient as possible with our resources we should hence make sure that we are not using superfluous parameters. We then would like some classical optimisation algorithm for parametrised quantum circuits that reduces the circuit to a form that has the minimal number of parameters, while still being able to express the same set of unitaries.

Parameter optimisation is also relevant in the field of measurement-based quantum computation (MBQC)~\cite{MBQC1,MBQC2}. In most literature on MBQC the measurement patterns have to be deterministic regardless of the chosen measurement angles~\cite{GFlow,DP2,wetering-gflow}. We can hence view this as a computation parametrised by those measurement angles. In this setting minimising the number of parameters corresponds to minimising the number of measurements needed~\cite{wetering-gflow}, while still preserving deterministic realisability of the pattern.

There are many results in general circuit optimisation, that try to simplify the presentation of quantum circuits in terms of total number or gates or circuit depth, but also in trying to reduce the count of a specific gate, like the number of CNOTs or T gates. Several of these techniques only include special-case behaviour for phase gates that are Clifford, like the S or Z gate, but not for other types of phases~\cite{nam2018automated,cliffsimp,kissinger2019tcount,meet-in-the-middle2013,zhang2019optimizing}. This means that these techniques also apply to parametrised phase gates, as we can then just treat these as `black-box' non-Clifford phase gates. 

There are only a small number of rewrites known that we can do with such black-box phase gates. Most of these correspond to a `fusing' of two parameters.  The simplest case is when we have two parametrised phase gates next to each other on the same qubit:
\begin{equation}\label{eq:circuit-phase-fuse}
  \vcenter{\hbox{\includegraphics[height=1.4\baselineskip]{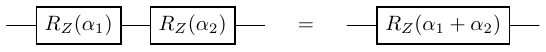}}}
\end{equation}
We see that the two different parameters $\alpha_1$ and $\alpha_2$ are combined into one gate, so that we can build the same set of unitaries using a single parameter $\alpha'$ by making the identification $\alpha' := \alpha_1+\alpha_2$.
More complicated versions of the same basic idea can be made by using the sum-of-paths approach~\cite{AmyVerification}, which allows us to fuse phases that apply to the same parity of qubits~\cite{meet-in-the-middle2013}, or by compiling the circuit into a series of Pauli exponentials and exploiting the Pauli commutation relations~\cite{kissinger2019tcount,zhang2019optimizing}. We will refer to these techniques collectively as \emph{phase folding}.

There are also more complicated rewrites that can in principle be performed on parametrised phase gates, for instance by exploiting the different Euler angle decompositions of a single-qubit unitary:
\begin{equation}\label{eq:generalised-euler}
  \vcenter{\hbox{\includegraphics[height=1.4\baselineskip]{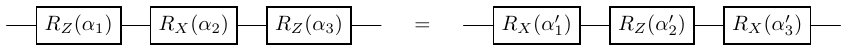}}}
\end{equation}
Here the parameters on the right depend on those on the left by some trigonometric relations,
and are in particular discontinuous as a function of $(\alpha_1, \alpha_2, \alpha_3)$. Hence, when the use-case is for instance QAOA, this might not be desirable to apply, as it transforms the parameter space in pathological ways.

All this then raises a number of questions:
\begin{itemize}
  \item What is the right notion of equivalence of parametrised quantum circuits?
  \item What are the possible rewrites we can do to transform parametrised quantum circuits while preserving equivalence?
  \item Is there an efficient algorithm to find an equivalent parametrised quantum circuit that uses the minimal number of parameters?
\end{itemize}
In this work, we answer each of these questions in the case of parametrised quantum circuits without repeated parameters.
Regarding the first point, we show that a broad set of parameter transformations (analytic functions on the unit circle) actually already forces the relations between parameters in equivalent quantum circuits to be given by simple additive relations that correspond to just `fusions' of parameters. 
We answer the second point by finding that the Clifford rewrite rules of the ZX-calculus~\cite{BackensCompleteness} suffice to prove any equality between parametrised Clifford circuits, under the condition that each parameter occurs uniquely.

Our main result is an answer to the third question: we find that an existing optimisation approach, previous work of some of the authors~\cite{kissinger2019tcount}, finds the optimal number of parameters, under the condition that every parameter in the circuit occurs on a unique gate. More formally we show the following:
\begin{theorem}
  Given a parametrised circuit which consists of Clifford gates and parametrised $Z$ phase gates each of which is parametrised by a unique parameter, we can efficiently find an equivalent parametrised circuit with an optimal number of parameters. Furthermore, these new parameters correspond to sums and differences of the original parameters, and the algorithm for finding the circuit is that described in~\cite{kissinger2019tcount}.
\end{theorem}
Note that if we drop the requirement here on the non-parametrised gates being Clifford that the problem likely no longer has an efficient solution: we show that when Clifford+T gates are allowed, parameter optimisation is NP-hard. We conjecture that the same is true when parameters are allowed to be used multiple times on different gates, and in fact we show this is the case for optimising \emph{post-selected} quantum circuits with repeated parameters.

\section{Parametrised circuits}

We will consider a parametrised quantum circuit to be a quantum circuit consisting of gates from a discrete set, together with \emph{parametrised phase gates} $Z[\alpha]$. Here $\alpha\in \R$ is a classical parameter that has to be determined prior to running the quantum circuit. We consider such a parameter to be a complete unknown which can be in any value in $\R$. A parametrised quantum circuit $C$ depending on a vector of values $\vec \alpha \in \R^k$ can then be viewed as a map from parameter space $\R^k$ to the space of unitaries $\mathcal{U}$ as $C: \R^k \to \mathcal{U}$. A specific instantiation of $C$, which we will write as $C[\vec \alpha]$, is then just a regular quantum circuit where the phases have been filled in.

In order to consider optimising parametrised quantum circuits, we first need to discuss what it means for parametrised circuits to be equal. 
\begin{definition}\label{def:reduction}
  We say a circuit $C_1$ with parameter space $\R^k$ \emph{reduces} to $C_2$ with parameter space $\R^l$ when there exists a function $f:\R^k\to \R^l$ such that $C_1[\vec \alpha] = \lambda(\vec \alpha) C_2[f(\vec \alpha)]$ for all $\vec \alpha\in \R^k$ where $\lambda: \R^k\to \CC\setminus \{0\}$ is a function representing a global scalar $\lambda(\vec \alpha)$ that may depend on the parameters. 
\end{definition}
The notion of reduction will be sufficient for us, but a natural stronger property to consider would be \emph{equivalence}, which would require a reduction to exist in both directions, essentially stating that both circuits are equally expressive.

We can then define the main problem we will study in this paper:
\begin{definition}
  \textbf{Parameter optimisation}: Given a parametrised quantum circuit $C_1$ with parameter space $\R^k$, find a parametrised quantum circuit $C_2$ with parameter space $\R^l$ that it reduces to, such that $l$ is \emph{minimal}.
\end{definition}

Without any further restriction on $C_1$, $C_2$ and $f$ this problem is likely to be hard, and so we don't expect an efficient solution. In particular, if the discrete gate set we are allowed to use is approximately universal, then parameter optimisation is likely to be hard. Consider for instance the following result regarding parametrised circuits with Clifford+T gates.
\begin{proposition}\label{prop:hardness-clifford-T}
  Parameter optimisation when the circuits are allowed to contain Clifford+T gates is NP-hard.
\end{proposition}
\begin{proof}
  We use an argument similar to~\cite{wetering2023optimising}. When we allow T gates in our circuit we can construct classical oracles $U_f\ket{\vec x, y} = \ket{\vec x, f(\vec x)\oplus y}$ for a Boolean function $f:\{0,1\}^n \to \{0,1\}$. Using parametrised phase gates $Z(\alpha)$ and $Z(\beta)$ we can then build a circuit implementing the diagonal unitary $\ket{\vec x, y} \mapsto e^{i\alpha y + i\beta (f(\vec x)\oplus y)}\ket{\vec x,y}$, by first applying $Z(\alpha)$ on $\ket{y}$, then computing with $U_f$, applying $Z(\beta)$, and uncomputing $U_f$. In this case it is clear that the parameters $\alpha$ and $\beta$ can be fused if and only if $f$ is never satisfiable or always satisfiable. Hence, an oracle for finding the minimal number of parameters for Clifford+T circuits with parametrised phase gates would allow us to solve the NP-complete SAT problem with a little bit of post-processing.
\end{proof}

For this reason we will restrict to just Clifford gates as our allowed discrete gates. But even with just Clifford gates and parametrised phase gates we need to be careful to not make the problem too hard. For instance, if we were to allow the mapping between the parameter space and the concrete phase parameters that get used in the phase gates to be non-identity, then we also quickly run into problems. Consider for instance two parametrised phase gates on the same qubit, which depend on $\alpha$ via $Z(f_1(\alpha))$ and $Z(f_2(\alpha))$. Then if we take $f_1(\alpha) = \alpha$ and $f_2(\alpha) = -\alpha + \frac\pi4$, then these two phase gates fuse (Eq.~\eqref{eq:circuit-phase-fuse}) and we get a $Z(\frac\pi4) = T$ gate, so that optimising such a circuit is again NP-hard.

Hence, we work with circuits consisting of Clifford gates and parametrised phase gates $Z(\alpha)$ where $\alpha$ exactly corresponds to the value of our parameter space (we could have also allowed for additions of multiplies of $\frac\pi2$ in the phase gate, but these could just as well be part of the discrete gate set, so we don't allow that without loss of generality). In this setting we are still allowed to reuse parameters into multiple phase gates in the circuit. This allows us to for instance construct parametrised controlled-phase gates. We suspect however that this problem is also already hard.
\begin{quote}
  \textbf{Conjecture}: The parameter optimisation problem for Clifford circuits with parametrised phase gates where parameters are allowed to be used in multiple gates is NP-hard.
\end{quote}
While we have not managed to prove this, we have proven this result in the setting of post-selected quantum circuits:
\begin{proposition}
  Parameter optimisation for post-selected Clifford quantum circuits with repeated parametrised phase gates is NP-hard.
\end{proposition}
The proof can be found in Appendix~\ref{app:repeated-params}. While it might seem obvious that a problem involving post-selected circuits would be hard, we point out that our main result about optimally minimising parameter count for Clifford+parameter circuits where the parameters occur uniquely \emph{also} applies to post-selected circuits (see Remark~\ref{rem:post-selection}). Hence, the optimisation problem is efficiently solvable for post-selected circuits when each parameter occurs uniquely, but it is NP-hard when the parameters are allowed to repeat.

While our proof of hardness does not extend to unitary circuits, we still expect the problem to be hard. We note for instance that for unitary circuits with repeated parametrised phase gadgets there are additional rewrites of the circuits we should consider that do not play a role in the case where every parameter is used once, such as: 
\begin{equation}
  \input{controlled-phase-identity.tikz}
\end{equation}
I.e.~this says that a controlled Z-phase gate commutes with an anti-controlled X-phase gate. This rewrite rule was needed to prove completeness for universal ZX-diagrams in~\cite{SimonCompleteness}, and hence seems to be a `harder' rule then the ones we need. This rewrite rule is not covered by our optimality result, since the ZX-diagram of a controlled-phase gate requires the parameter to be used three times~\cite{CD1}.

This discussion on the hardness of the problem then brings us to the types of circuits we will consider in this paper.
\begin{definition}
  A \emph{parametrised Clifford circuit} (PCC) is a parametrised circuit consisting of Clifford gates and $Z[\alpha]$ parametrised phase gates where each parameter is used in at most one gate.
\end{definition}
One way to interpret this condition on each parameter occurring uniquely is that all the parametrised phase gates are truly `black boxes': we know nothing about their values, not even whether two of the parametrised phase gates are implementing the same phase. In any case it is clear that any circuit with repeated parameters reduces to one where the parameters occur uniquely: for every repeat occurrence we instead introduce a fresh parameter. Note that this setting of unique parameters is also a natural setting to consider when talking about determinism in measurement-based quantum computing; see Section~\ref{sec:MBQC}.

We will also restrict the types of functions $f$ that may appear in a reduction of a parametrised Clifford circuit. Instead of treating the parameters as real numbers, it makes sense to treat them as phases modulo $2\pi$, and hence as elements of the unit circle $S^1$. It then additionally makes sense to require the parameter map $f: (S^1)^k\to (S^1)^l$ to be continuous: many use-cases of parametrised quantum circuits require us to take derivatives and slightly change the parameters bit by bit. A reduction then shouldn't break these properties. In fact, we will go beyond that and assume that the parameter map is \emph{smooth}: infinitely differentiable. Since there also exist pathological smooth functions, we will make our final assumption, which is that parameter maps are \emph{analytic}, meaning they can be written as convergent infinite power series. Since the details of this are a bit technical to formalise, we present them in Appendix~\ref{app:affine-functions}. We just present here the main conclusion.
\begin{proposition}
  Any parameter map $\Phi:\R^k\to \R^l$ consisting of analytic functions on the unit circle (as formalised in Appendix \ref{app:affine-functions}) is equal to $\Phi(\vec \alpha) = M\vec \alpha + \vec c \text{ mod } 2\pi$ where $M$ is a matrix of integers and $\vec c \in \R^l$ is some constant.
\end{proposition}

Hence, we see that what a priori seems like quite a general class of transformations actually reduces to something quite simple.

\begin{definition}
  We say a PCC $C_1$ \emph{affinely reduces} to $C_2$ when for all $\vec \alpha$ we have $C_1[\vec \alpha] = \lambda(\vec\alpha) C_2[M \vec \alpha + \vec c]$ for some integer matrix $M$ and constant $\vec c$ and `global phase function' $\lambda(\vec\alpha)\in\mathbb{C}\setminus\{0\}$.
\end{definition}

Note that when we have an affine reduction $C_1[\vec \alpha] = C_2[M \vec \alpha + \vec c]$, that a parameter $\alpha_j$ gets copied to all the parametrised gates in $C_2$ where the $j$th column of $M$ is non-zero. Hence, such an affine reduction can break the condition on having the parameter only appear in one place. For this reason we define a stricter version of reduction.

\begin{definition}
  We say an affine reduction is \emph{parsimonious} when the matrix $M$ appearing in the reduction has at most 1 non-zero element in each column.
\end{definition}

\begin{definition}
  The \textbf{affine parsimonious parameter optimisation problem for Clifford circuits} is to find, given a parametrised Clifford circuit $C_1$, a parametrised Clifford circuit $C_2$, such that $C_1$ parsimoniously affinely reduces to $C_2$ and $C_2$ has a minimal number of parameters amongst those circuits that $C_1$ parsimonious affinely reduces to.
\end{definition}

The main result of this paper is then the following theorem.
\begin{theorem}
  There is an efficient algorithm that constructively solves the affine parsimonious parameter optimisation problem for Clifford circuits.
\end{theorem}
This algorithm is the T-count optimisation algorithm of~\cite{kissinger2019tcount}. In that paper it was already noted that their algorithm can be used to reduce the number of parameters in a circuit. Our results show that it is in fact \emph{optimal} in doing so (given suitable restrictions on the type of reductions that are allowed as discussed above).

Note that the assumption of the reduction being parsimonious seems quite strong. We suspect however that this condition is not needed, and that for any parametrised Clifford circuit $D_1$ for which there is an affine reduction to $D_2$, that there is then also a parsimonious affine reduction to $D_2$. We have however only managed to prove this in certain special cases (see Section~\ref{sec:general-affine}).

\section{The ZX-calculus}

The algorithm that implements the optimal parameter reduction strategy uses the ZX-calculus~\cite{CD1,CD2}, and our proof is also most naturally expressed in the ZX-calculus. We will hence give a brief introduction to the ZX-calculus and the rewrite strategy of~\cite{kissinger2019tcount}. For an in-depth review, see~\cite{vandewetering2020zxcalculus}.

\noindent ZX-diagrams are built out of two types of generators, called \emph{spiders}. We have \emph{Z-spiders}:
\begin{equation}
  \input{Zsp-a.tikz} \ = \ \ketbra{0\cdots 0}{0\cdots 0} + e^{i\alpha} \ketbra{1\cdots 1}{1\cdots 1}
\end{equation}
And \emph{X-spiders}:
\begin{equation}
  \input{Xsp-a.tikz} \ = \ \ketbra{+\cdots +}{+\cdots +} + e^{i\alpha} \ketbra{-\cdots -}{-\cdots -}
\end{equation}
The phase $\alpha\in \R$ can be any real number and we take it modulo $2\pi$. When $\alpha=0$ we don't write it on the spider. Spiders can have any number of inputs, including zero, and any number of outputs, also including zero. In particular, the zero-input zero-output spider corresponds to just a number: $\begin{tikzpicture}
	\begin{pgfonlayer}{nodelayer}
		\node [style=white phase dot] (0) at (0, 0) {$\alpha$};
	\end{pgfonlayer}
\end{tikzpicture}
 = 1 + e^{i\alpha}$. Another relevant special case is the 1-input 1-output X-spider with a $\pi$ phase, which is the Pauli $X$ (the NOT gate).

We can compose spiders together either vertically, corresponding to taking the tensor products of their linear maps, and horizontally, corresponding to regular composition of linear maps. Such a composition of spiders is called a \emph{ZX-diagram}. When the phases on the spiders are allowed to be arbitrary, ZX-diagrams are \emph{universal}, meaning that they can represent any linear map from $\CC^{2^n}$ to $\CC^{2^m}$. If instead we restrict the phases to be multiples of $\frac\pi4$, then the diagrams are universal for linear maps we can construct from Clifford+T circuits including ancilla preparations and postselections. If we further restrict to multiples of $\frac\pi2$, we get Clifford circuits, stabiliser state preparations, and stabiliser projections. In particular, we say a spider is \emph{Clifford} when $\alpha$ is a multiple of $\frac\pi2$.
A ZX diagram is scalar when it has no input or output wires.

The usefulness of ZX-diagrams comes from the set of rewrite rules that we can apply to them that preserve their semantics as linear maps. First of all, the interpretation of the diagram only depends on the connectivity, and we can hence deform the diagram in the plane as we wish. This means we can treat the diagrams as undirected graphs. Second, there is a set of local rewrite rules that preserve the semantics; see Figure~\ref{fig:clifford-rules}. These rules suffice to prove any true equality between Clifford diagrams, and hence are \emph{complete} for the Clifford fragment~\cite{BackensCompleteness}. Note however that these rules do not prove all equalities between general diagrams where the phases are not restricted to multiples of $\frac\pi2$.

\begin{figure}[!tb]
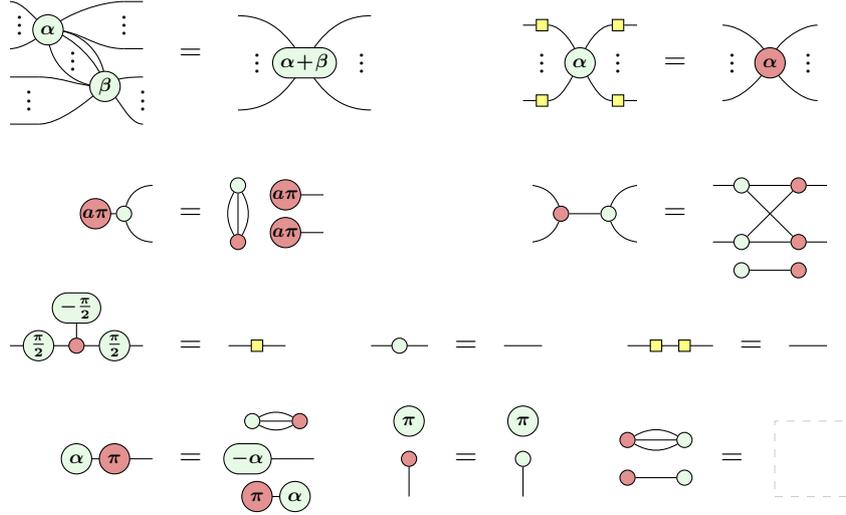

  \ctikzfig{zx-rules-Clifford}
  \caption{A complete set of rules for the Clifford fragment of the ZX-calculus. Here $a\in \{0,1\}$ is a Boolean variable, and the equations involving $\alpha$ or $\beta$ hold for any $\alpha,\beta\in \R$, not just the Clifford angles.}
  \label{fig:clifford-rules}
\end{figure}
 
Here we drew a yellow square for the \emph{Hadamard gate}, for which we can derive the equation on the left in Figure~\ref{fig:clifford-rules} as its definition.

Note that most of these rules only apply when the spiders have specific phases, except for \emph{spider fusion} (top left), \emph{colour change} (top right) and \emph{$\pi$-pushing} (bottom left), which are sound for any value $\alpha,\beta\in \R$. In particular $\pi$-pushing will be relevant to us, so let us note that we can also directly write the scalar diagrams as numbers to present this rule as:
\begin{equation}\label{eq:pi-push}
  \input{pi-push.tikz}
\end{equation}
This just says that $X (\ket{0} + e^{i\alpha} \ket{1}) = e^{i\alpha} (\ket{0} + e^{-i\alpha} \ket{1})$.

Because we can treat ZX-diagrams as undirected graphs and the rewrite rules apply regardless of which wires are inputs or outputs, we will in this paper often work with \emph{states} for convenience: diagrams that have no inputs. This corresponds to using the Choi-Jamio\l{}kowski isomorphism, and in the diagram is represented by bending the input wires to be outputs instead. When we introduce parametrised diagrams in Section~\ref{sec:parameters-ZX} we will usually denote the parameters as being connected to inputs, so that it is a `parametrised state'.

\subsection{Local complementation and pivoting}

For this paper it is important that we can bring Clifford diagrams to particular pseudo-normal forms, so we will review some useful rewriting tools that we need to reduce to these normal forms.

First, instead of working with arbitrary ZX-diagrams, it is often helpful to restrict to \emph{graph-like} ZX-diagrams (which can be done without loss of generality). These are diagrams where all spiders are Z-spiders, and every spider is `maximally fused', meaning that the only connections left are via Hadamard gates~\cite{cliffsimp}. Parallel edges and self-loops can be removed, so that we are left with a simple graph where the vertices are the Z-spiders and the edges correspond to Hadamards, which we will refer to as \emph{Hadamard edges} from now on. Each vertex is additionally labeled by the phase of the spider. For clarity we will write Hadamard edges as blue-dotted wires:
\begin{equation}
  \input{blue-edge-def.tikz}
\end{equation}

The spiders that are connected to an input or output wire of the diagram we will refer to as \emph{boundary} spiders, while the spiders that are only connected to other spiders we will call \emph{internal} spiders. The rewrite strategy of~\cite{cliffsimp} describes how to remove all the internal Clifford spiders.
First, the \emph{local complementation} rewrite rule removes any internal spider with a phase of $\pm\frac\pi2$:
\begin{equation}\label{eq:lc-simp}
  \input{lc-simp.tikz}
\end{equation}
Note that this equation is actually only accurate up to some scalar value, depending on the number of neighbours $n$ and the phase $\pm\frac\pi2$, but as this value is representable by a Clifford diagram it will not be important to us.

The second equation is called the \emph{pivot} rule, and it can remove any pair of connected internal spiders that both have a $0$ or $\pi$ phase:
\begin{equation}\label{eq:pivot-simp}
  \input{pivot-simp-smaller.tikz}
\end{equation}
Hence, if we started with a Clifford diagram, one where all the phases are multiples of $\frac\pi2$, then the only internal spiders left, after removing those with a $\pm\frac\pi2$ phase with Eq.~\eqref{eq:lc-simp} and the remaining ones with a $0$ or $\pi$ phase that are connected, are those that have a $0$ or $\pi$ phase and are not connected to any other internal spider. Assuming for simplicity that the diagram only has outputs and no inputs, we can write such a diagram in two layers, with first the internal spiders, and then the boundary spiders:
\begin{equation}
  \input{ap-form-ex.tikz}
\end{equation}
We say such a diagram is in \emph{affine with phases} form, or AP form. This is because we can see such a diagram as encoding an affine subspace, together with a phase polynomial~\cite{poor2023qupit,mcelvanney2022complete}:
\begin{equation}\label{eq:AP-state-decompose}
   \input{AP-state-decompose.tikz}  
\end{equation}

Using an additional rewrite we can also get rid of the final internal spiders, at the cost of introducing Hadamard edges on the output wires. We call this rewrite rule a \emph{boundary pivot}, as it applies a pivot between an internal spider and a boundary spider:
\begin{equation}\label{eq:pivot-boundary}
    \input{pivot-boundary-simp.tikz}
\end{equation}
This introduces a new element to the diagram that we call a \emph{phase gadget}, that we will have more to say about in the next section. But for now note that if $\gamma=\pm\frac\pi2$, that we can remove it with a local complementation, after which its neighbour will also have a $\pm\frac\pi2$ phase so that it can also be removed. If instead $\gamma$ is $0$ or $\pi$ then we can remove both of these spiders with a regular pivot. We then see that in both cases we can again remove both introduced spiders. Hence, by repeating the boundary pivot we can get rid of the final internal spiders. Hence, when we start with a Clifford diagram, we can efficiently reduce to a pseudo-normal-form where there are no internal spiders, but there may be Hadamard gates on the output wires:
\begin{equation}
  \input{graph-state-example.tikz}
\end{equation}
This is a graph state with local Cliffords (GSLC). But note that these local Cliffords are from a small set. They are either a power of an $S$ gate (a $k\frac\pi2$ phase), or a Hadamard gate, or a $Z$ gate (a $\pi$ phase) followed by a Hadamard (an $S$ gate followed by a Hadamard does not occur, since Hadamards on a boundary only get introduced by boundary pivots).

\subsection{Optimising non-Clifford diagrams}\label{sec:opt-non-clifford}

The above rewrite strategy can also be applied when we have non-Clifford phases in the diagram. In that case local complementation still removes all internal spiders with a $\pm\frac\pi2$ phase, and pivoting still removes all connected internal spiders with a $0$ or $\pi$ phase. If an internal $j\pi$ phased spider is connected to some boundary spider with a Clifford phase then we can also get rid of it with a boundary pivot as described above. However, now we can have internal $0$ or $\pi$ spiders that are connected only to other spiders with a non-Clifford phase. In that case we can't remove the spider, but we can make it part of a \emph{phase gadget}:
\begin{equation}
  \input{phase-gadget.tikz}
\end{equation}
Namely, if an internal $j\pi$-phased spider is connected to a non-Clifford spider, we apply a \emph{gadget pivot}:
\begin{equation}\label{eq:gadget-simp}
  \input{pivot-gadget-simp.tikz}
\end{equation}
If instead it is only connected to boundary non-Clifford spiders, then we apply the boundary pivot of Eq.~\eqref{eq:pivot-boundary}, but now we don't remove the phase gadget. Note that in both cases the phase gadget ends up connected to what the $j\pi$ spider was connected to originally. As we assumed that this spider was only connected to non-Clifford spiders, we see that when we repeat this procedure that phase gadgets never end up connected to each other.

After we have done these operations wherever we can, we see that the only internal spiders have a non-Clifford phase or are part of a phase gadget with a non-Clifford phase. We can then write this new diagram also as a GSLC, but now where some of the outputs are plugged with a non-Clifford phase:
\begin{equation}
  \input{graph-state-plugged.tikz}
\end{equation}
Hence, we see that the phase gadgets correspond to an output where there is a Hadamard and an effect corresponding to the phase, while for internal non-Clifford spiders we can view them as appearing as effects on output wires where there is no Hadamard in between.

Note that none of these steps so far actually remove any of the non-Clifford phases. To do that, we use \emph{gadget fusion}~\cite{kissinger2019tcount}:
\begin{equation}\label{eq:gadget-fusion}
  \input{phase-gadget-fusion.tikz}
\end{equation}
This applies whenever two phase gadgets have exactly the same set of neighbours, and it fuses them into one phase gadget. This rule captures the aformentioned concept of phase folding. The other rule that combines non-Clifford phases applies instead to a phase gadget that has exactly one neighbour:
\begin{equation}\label{eq:gadget-id}
  \input{phase-gadget-id-fuse.tikz}
\end{equation}
Additionally, any internal spider or phase gadget that has no neighbours only contributes a scalar to the overall diagram and hence can also be removed. Note that these fusions of non-Clifford phases can result in a spider that has a Clifford phase. For instance, two $\frac\pi4$ phases combine into a $\frac\pi2$ phase. In that case we can us the previously described rewrites to remove this newly acquired Clifford spider.

The set of rewrites we have described here and the rewrite strategy it implies was proposed in~\cite{kissinger2019tcount}. To use it to optimise diagrams, in~\cite{kissinger2019tcount} they used the technique of \emph{phase teleportation} to let it inform when phases in the original circuit were allowed to be fused. Later, in~\cite{wetering-gflow} it was shown that a quantum circuit can also be directly extracted from the diagram using techniques from measurement-based quantum computing. This procedure preserves the number of non-Clifford phases in the diagram.

As it will be important for us later, let us describe the type of pseudo-normal-form the rewriting procedure from~\cite{kissinger2019tcount} gives: we get a graph state where on each output we have either a $k\frac\pi2$ phase, or a $Z^aH$ unitary where $a\in\{0,1\}$, and some of the outputs can be plugged by a non-Clifford phase. Those plugged spiders with a Hadamard on them are phase gadgets and have at least 2 neighbours. Each phase gadget additionally has a unique set of neighbours (since if two gadgets had the same neighbours, we would have fused them).

% We described the above procedure for arbitrary diagrams. If we started out with a unitary circuit, then the diagram we get in the end doesn't necessarily looks like a circuit. It is described in~\cite{wetering-gflow} how to transform the resulting diagram back into a circuit. This procedure preserves the number of non-Clifford phases in the diagram.

\section{Parametrised circuits in the ZX-calculus}\label{sec:parameters-ZX}
In order to work with parametrised circuits in the ZX-calculus we will define the notion of a parametrised diagram.

\begin{definition}
  A \emph{parametrised diagram with $k$ parameters} $D$ is a ZX-diagram that is entirely Clifford, apart from the phases $\vec \alpha = (\alpha_1,\ldots, \alpha_k)$ that are each present \emph{at most once} in the diagram. That is, for each choice of $\alpha_j\in \R$ we have
  \[\input{D-alpha.tikz}\]
  where the equality is on the nose (so with accurate scalars), and $D'$ is a Clifford diagram.
  We wrote the diagram here as a state (no input wires) without loss of generality. 
  Note that if some phase $\alpha_j$ is actually not present in $D[\alpha_1,\ldots \alpha_k]$ that we can still find a Clifford diagram $D'$ such that the above holds, as \input{scalar-alpha-cancel.tikz}.
  We will write $D$ for the ZX-diagram with `abstract' parameters, and $D[\vec \alpha]$ for the diagram instantiated with a particular choice of values $\alpha_1,\ldots \alpha_k$.
\end{definition}

The condition here that each phase only appears at most once is important, as it means that on the input wires of $D'$, each has a phase we can choose independently. If some phase $\alpha_j$ appeared multiple times in $D$, then to `extract' it out of the Clifford portion would mean that we have a number of input wires that are `correlated'. As a result we would only have information about the symmetric subspace of these wires.

\begin{remark}
  Scalar factors will be important in this paper. Every equality written in this section will hence be meant to be \emph{on the nose}, meaning that they represent equal linear diagrams, even with the correct global phase.
\end{remark}

As we will want to consider reductions between parametrised diagrams, it will be helpful to first note a couple of properties regarding equalities between parametrised diagrams and what the ZX-calculus can prove about these.

\subsection{Completeness of the Clifford+parameters fragment}

We can straightforwardly prove that the ZX-calculus is complete for parametrised Clifford diagrams.

\begin{definition}
  Write $\ket{+_\alpha} = \begin{tikzpicture}
	\begin{pgfonlayer}{nodelayer}
		\node [style=Z phase dot] (15) at (0, 0) {$\alpha$};
		\node [style=none] (16) at (1.25, 0) {};
	\end{pgfonlayer}
	\begin{pgfonlayer}{edgelayer}
		\draw (15) to (16.center);
	\end{pgfonlayer}
\end{tikzpicture}
 = \ket{0} + e^{i\alpha} \ket{1}$.
\end{definition}
Note that in this definition we use unnormalised states. Hence $\ket{+_0} = \sqrt{2}\ket{+}$.

\begin{lemma}\label{lem:basis-decomp}
  Let $\alpha,\beta\in \R$ be two different phases (modulo $2\pi$). Then $\ket{+_\alpha}$ and $\ket{+_\beta}$ form a basis.
  In particular, writing $\ket{+_\gamma} = a \ket{+_\alpha} + b \ket{+_\beta}$ we have $a+b = 1$ and 
  \begin{equation}
    a \ = \ \frac{e^{i\gamma} - e^{i\beta}}{e^{i\alpha} - e^{i\beta}} \qquad \quad b \ =\ \frac{e^{i\alpha} - e^{i\gamma}}{e^{i\alpha} - e^{i\beta}} 
  \end{equation}
  In the special case of $\alpha=0$ and $\beta=\pi$ we have $a-b = e^{i\gamma}$.
\end{lemma}

\begin{proposition}\label{prop:parameter-equality}
  The rules of Figure~\ref{fig:clifford-rules} are complete for parametrised Clifford diagrams:
  Let $D_1$ and $D_2$ be two parametrised diagrams with the same number of parameters such that $D_1[\vec \alpha] = D_2[\vec \alpha]$ for every choice of $\alpha_1,\ldots, \alpha_k$. Then $D_1' = D_2'$ and $D_1[\vec \alpha]$ and $D_2[\vec \alpha]$ can be uniformly rewritten into each other using the ZX-calculus Clifford rewrite rules.
\end{proposition}
\begin{proof}
  Let $\ket{\psi} = a\ket{+_0} + b \ket{+_\pi}$ be an arbitrary (not necessarily normalised) qubit state. Note that we can write $\ket{\psi} = a~\begin{tikzpicture}
	\begin{pgfonlayer}{nodelayer}
		\node [style=white dot] (15) at (0, 0) {};
		\node [style=none] (16) at (1.25, 0) {};
	\end{pgfonlayer}
	\begin{pgfonlayer}{edgelayer}
		\draw (15) to (16.center);
	\end{pgfonlayer}
\end{tikzpicture}
 + b~\begin{tikzpicture}
	\begin{pgfonlayer}{nodelayer}
		\node [style=white phase dot] (15) at (0, 0) {$\pi$};
		\node [style=none] (16) at (1.25, 0) {};
	\end{pgfonlayer}
	\begin{pgfonlayer}{edgelayer}
		\draw (15) to (16.center);
	\end{pgfonlayer}
\end{tikzpicture}
$. Hence:
  \[\input{D-psi-plug-pf.tikz}\]
  As these two linear map agree on whatever state we put in, they must be equal:
  \[\input{D-psi-plug-pf2.tikz}\]
  We can repeat this procedure for all the other $\alpha_j$ to conclude that $D_1' = D_2'$. As these are both Clifford diagrams and they represent the same linear map, by Clifford completeness there is a sequence of Clifford rewrites that transforms $D_1'$ into $D_2'$. When we plug in some $\alpha_j$ phases into these diagrams, this set of rewrites will still be applicable.
  Hence, there is a set of Clifford rewrites transforming $D_1[\vec \alpha]$ into $D_2[\vec \alpha]$.
\end{proof}

\begin{remark}\label{rem:finite-values}
  In this result we didn't have to use the fact that these diagrams were equal for any choice of $\alpha$ here. It suffices for the diagrams to agree on just two different values of $\alpha_j$ for each $j$. 
  This is because $\{\ket{+_\alpha},\ket{+_\beta}\}$ forms a basis of the qubit state space if $\alpha\neq \beta$. This is a special case of the possibility to verify parametrised equations using a finite number of cases of~\cite{MillerBakewell2020finite}. In particular, if we had repeated parameters, then if a parameter was repeated $k$ times, we would have to verify equality of the diagrams at $k+1$ different values of the parameter. This is because the symmetric subspace of $k$ qubits is $(k+1)$-dimensional.
\end{remark}

\begin{remark}
  Note that these completeness results require each parameter to appear freely in the diagram (i.e.~with no dependencies between them), and hence that we can't have repeated parameters. There are equations involving repeated parameters that cannot be proven by the Clifford rewrite rules, such as supplementarity~\cite{supplementarity}. In~\cite{JPV-universal}, it is shown that when we allow $\frac\pi4$ phases, so that the diagrams correspond to the Clifford+T fragment, that we can then prove any equation involving repeated parameters (what they call `linear diagrams with constants in $\frac\pi4$'), if we use the extended complete axiomatisation for Clifford+T ZX-diagrams. In this setting we however don't expect that rewriting is efficient in general, as the completeness result of~\cite{JPV-universal} requires rewriting to an exponentially large normal form.
\end{remark}

In Proposition~\ref{prop:parameter-equality} we had to assume that the diagrams were equal on the nose. But when we are dealing with unitary parametrised circuits, we only care about the exact unitary up to global phase. A stronger result would hence be to be able to rewrite diagrams that are only equal up to a global phase, where this phase may depend on the value of the parameters. We hence care about the situation where $D_1[\vec \alpha] = \lambda(\vec \alpha) D_2[\vec \alpha]$ for all $\vec \alpha$ and $\lambda(\vec \alpha)$ is a phase depending on the value of $\alpha$. We will work in a slightly more general setting then that and assume that $\lambda:\R^k\to \CC$ is just any function mapping phases into a scalar value (which is also allowed to be zero a priori).
We can then prove an analogous result to Proposition~\ref{prop:parameter-equality}, but to do that we will need to make an assumption on the parameters.

\begin{definition}
  Let $D$ be a parametrised diagram with parameters $\alpha, \vec \alpha$. We say the parameter $\alpha$ is \emph{trivial}, if there is a choice of Clifford angles $\vec \alpha \in \Z[\frac\pi2]^k$ and two distinct choices of value $\beta$ and $\gamma$ of $\alpha$ such that the diagrams $D[\beta, \vec \alpha]$ and $D[\gamma, \vec \alpha]$ are proportional to each other. Otherwise we call the parameter \emph{non-trivial}.
\end{definition}
This definition is motivated by the following result.
\begin{lemma}
  Suppose $\alpha$ is a trivial parameter in a parametrised diagram $D[\alpha,\vec \alpha]$. Then $D[\beta,\vec \alpha]$ is proportional to $D[\gamma, \vec \alpha]$ for every choice of $\beta,\gamma$ where $\vec \alpha$ is some Clifford choice of the other parameters.
\end{lemma}
\begin{proof}
  Let $\beta$ and $\gamma$ be two choices for the trivial parameter $\alpha$ such that the associated diagrams are proportional to each other. Write $D[\beta,\vec \alpha] = \lambda D[\gamma,\vec\alpha]$ where $\lambda$ denotes the proportionality scalar. 
  Let $\delta$ be any other phase and let $a$ and $b$ be the values such that $\ket{+_\delta} = a \ket{+_\beta} + b \ket{+_\gamma}$. These values exist because $\ket{+_\beta}$ and $\ket{+_\gamma}$ form a basis.
  Then $D[\delta,\vec\alpha] = a D[\beta,\vec\alpha] + b D[\gamma,\vec\alpha] = (a\lambda + b) D[\gamma,\vec\alpha]$. Hence, the diagram for $\delta$ is proportional to that of $\gamma$.
\end{proof}

This lemma says that if a diagrammatic parameter does not change the diagram for two choices of the parameter, then it does not change the diagram for \emph{any other} choice of that parameter as well. Note that the converse says that if a parameter is \emph{non-trivial}, then that must mean that for any choice of the other parameters, changing the parameter must result in a different linear map.
We primarily care about diagrams coming from a quantum circuit, in which case every parameter is necessarily non-trivial.

\begin{proposition}\label{prop:circ-nontrivial}
  Let $D$ be a parametrised diagram that comes from a quantum circuit which contains parametrised phase gates. Then every parameter is non-trivial.
\end{proposition}
\begin{proof}
  Pick some parameter $\alpha$. As $D$ comes from a quantum circuit we can then write it as $D = U_2 \circ Z_1[\alpha]\circ U_1$ for some unitaries $U_1$ and $U_2$ (built out of the gates appearing before and behind the $Z[\alpha]$ gate), and where we write $Z_1[\alpha]$ as shorthand for the $Z[\alpha]$ gate appearing on the first qubit. Note that while $U_1$ and $U_2$ can depend on the other parameters, they do not depend on $\alpha$ as the parameter only appears in one gate. Let $D[\alpha]$ and $D[\alpha']$ be the diagram instantiated with the specific choices of $\alpha$ and $\alpha'$. We then claim that $D[\alpha]\circ D[\alpha']^\dagger$ is not equal to the identity. To see this note that $D[\alpha]\circ D[\alpha']^\dagger = U_1\circ Z_1[\alpha - \alpha'] \circ U_1^\dagger$. Hence, if this were equal to the identity, then 
  $$\text{id} \ =\  U_1^\dagger \circ \text{id} \circ U_1 \ =\  U_1^\dagger \circ U_1 \circ Z_1[\alpha - \alpha'] \circ U_1^\dagger\circ U_1 \ =\  Z_1[\alpha - \alpha'].$$
  But we know that this is not the identity when $\alpha \neq \alpha'$.
\end{proof}

In Appendix~\ref{app:completeness} we prove the following result that shows that when two diagrams are equal up to some scalar that may vary depending on the parameters, that actually this scalar does \emph{not} depend on the parameters and is just some constant.

\begin{proposition}\label{prop:scalar-is-constant}
  Let $D_1$ and $D_2$ be (not necessarily Clifford) parametrised diagrams with the same number of parameters and suppose $D_1[\vec \alpha] = \lambda(\vec \alpha) D_2[\vec \alpha]$ for all $\vec \alpha$ for some non-zero scalar function $\lambda(\vec \alpha)\in \CC$, and suppose that all the parameters of $D_1$ are non-trivial. Then $D_1' = \lambda' D_2'$ for some constant scalar $\lambda'$.
\end{proposition}

Using this, it is then straightforward to get the completeness result for parametrised Clifford diagrams, allowing equality to vary per parameter.

\begin{theorem}\label{thm:completeness-scalars}
  Let $D_1$ and $D_2$ be two parametrised diagrams with the same number of parameters and where all the parameters of $D_1$ are non-trivial such that $D_1[\vec \alpha] = \lambda(\vec \alpha) D_2[\vec \alpha]$ for all $\vec \alpha$ for some scalar function $\lambda(\vec \alpha)\in \CC$. Then $D_1' = C\otimes D_2'$ for some Clifford scalar $C$ and we can uniformly rewrite $D_1[\vec \alpha]$ into $C\otimes D_2[\vec \alpha]$ using Clifford rewrites. 
\end{theorem}
Here by `uniformly rewriting' we mean that the path of rewrites does not depend on the specific parameter inputted into the parametrised diagram, but is `parameter agnostic'.

\subsection{Optimising parametrised Clifford circuits}\label{sec:opt-parameters}

The optimisation strategy for diagrams containing non-Clifford phases of Section~\ref{sec:opt-non-clifford} applies equally well to parametrised Clifford diagrams. We just treat spiders that contain a phase depending on a parameter as a non-Clifford. The local complementation rewrite Eq.~\eqref{eq:lc-simp} can add a $\pm\frac\pi2$ phase to a spider containing a parameter, while the pivot rewrite rule Eq.~\eqref{eq:pivot-simp} can add a $\pi$ phase.
The boundary pivot~\eqref{eq:pivot-boundary} and gadget pivot~\eqref{eq:gadget-simp} can flip the parameter $\alpha$ to $-\alpha$, and finally the gadget fusion~\eqref{eq:gadget-fusion} and neighbour fusion~\eqref{eq:gadget-id} can add together the phases in two spiders so that we can end up with expressions like $\alpha+\beta$ on a single spider. 

We see then that if we started with parameters $\alpha_1,\ldots,\alpha_k$ that we then end up with spiders that can have an expression like $b\frac\pi2 + (-1)^{a_1}\alpha_{j_1}+\cdots+(-1)^{a_l}\alpha_{j_l}$ for some different indices $j_1,\ldots,j_l$. We will then label this expression by a new parameter name $\beta$. Note that since none of the rewrite rules copy the phase of a non-Clifford spider that each parameter appears uniquely on some spider, and hence the new parameters $\beta_1,\ldots, \beta_{k'}$ we get are all independent. 
In particular, the new parameters $\vec \beta$ are given by an affine transformation $\vec \beta = P \vec \alpha + \vec c$, where $\vec c = (b_1\frac\pi2,\ldots,b_{k'}\frac\pi2)$ is a vector of Clifford phases, $P$ is a matrix only containing entries in $\{0,1,-1\}$, and each column of $P$ contains a single non-zero element (corresponding to each $\alpha_j$ appearing in a unique location).

There are now two ways to proceed with the reduced diagram. As described at the end of Section~\ref{sec:opt-non-clifford}, we can extract a circuit directly from the optimised diagram, resulting in a circuit with an equal number of non-Clifford gates as in the optimised diagrams, which in our case means the circuit contains as many parameters as the optimised diagram. 
Alternatively, we can use the \emph{phase teleportation}, which uses the optimisation strategy to inform which phases could be fused, and then use that to optimise the original circuit `in place'. In terms of reductions, we can phrase this as follows. 
If we started with the parametrised circuit $D_1[\vec\alpha]$ with $k$ parameters, and we optimised it to a parametrised diagram $D_2$ with $l$ parameters, giving the parsimonious reduction $D_1[\vec\alpha] = \lambda(\vec\alpha)D_2[P\vec\alpha +\vec c]$, then we can get from this an `in place' parsimonious reduction $D_1[\vec\alpha] = \lambda'(\vec\alpha)D_1[P'\vec\alpha,0,\ldots,0]$, where we can leave $k-l$ parameters to be set to zero. We will see later in Proposition~\ref{prop:phase-teleportation} that we can in fact derive an `in place' parsimonious reduction from any parsimonious reduction, and that hence the phase teleportation procedure works for any such reduction.
% Alternatively, we can use the mapping in parameter space $P$ obtained during the simplification procedure to reassign the parameter values of the original circuit, which will result in a circuit with the same parameter count as the other method, but with the same structure as the original circuit.
% In \texttt{pyzx} this is achieved by applying the \texttt{teleport\_reduce} function to the original circuit.
% [TODO figure]
% We are then left with a diagram that is a `plugged' GSLC, where the only things plugged are our new parameters $\beta_j$:
% \begin{equation}
%   \tikzfig{graph-state-plugged-beta}
% \end{equation}
% It was established in~\cite{wetering-gflow} that if our starting diagram is a circuit with non-Clifford phase gates, that then this rewrite strategy produces a diagram that can efficiently be transformed back into a circuit with the same number of non-Clifford phases as in the optimised diagram. The paper~\cite{wetering-gflow} assumed these were unknown non-Clifford phases, but the same argument continues to hold for parametrised phases, as the circuit extraction algorithm is agnostic to the specific non-Clifford values.

Note that for the optimised diagram, the properties of the graph described in Section~\ref{sec:opt-non-clifford} carry over, and that in particular every parameter that is connected via a Hadamard to a spider has at least two neighbours, and that its neighbourhood does not match that of any other parameter connected via a Hadamard. This is important for the proof of optimality in the next section.

The sequence of rewrites that constitutes this algorithm is not unique. There are usually many choices of where to apply the local complementations and pivots that remove spiders, and these different choices lead to different final diagrams. However, as the results of the next section will show, regardless of the choice of rewrites, the algorithm finds the minimal number of parameters.

\section{Proving optimality}

We will now show that the algorithm for minimising parameter counts is optimal, in the sense that it is the best possible among the parsimonious affine reductions.

% The ZX diagram after the simplification procedure $D_2[\vec\beta]$ will often contain fewer parameterised spiders than the original diagram $D_1[\vec\alpha]$. The PyZX simplification procedure only combines non-Clifford spiders together using phase teleportation; it does not use trigonometric relations such as the Euler rule,
% nor does it use rules that involve the same parameter twice such as the supplementarity rule.
% Hence, it is easy to see that $\vec\beta = P\vec\alpha$, where each column of $P$ contains at most one 1 or -1. All-zero columns correspond to parameters that do not contribute to the original diagram. Working in this restricted setting, we can show that PyZX always achieves the optimal number of parameters after simplification.

% \subsection{Strengthening the linear combination assumption}

% The third assumption in subsection \ref*{the-assumptions} assumes the output parameters can be any real, linear combination of the input parameters. However, we can strengthen this assumption by showing that the coefficients of the linear relation are necessarily 1, $-1$, or 0, and that each input parameter is used at most once in the linear relation. This will be useful in the main theorem.

First, we show that if the parameters are non-trivial, then any parsimonious affine reduction can't multiply a phase $\alpha$ to something like $2\alpha$: only $\alpha$ or $-\alpha$ is allowed.

\begin{lemma}\label{lem:params-non-duplicate}
Let $D_1$ and $D_2$ be parametrised diagrams with respectively $k$ and $l$ parameters, and let $D_1[\vec\alpha] = \lambda(\vec \alpha) D_2[P\vec\alpha + \vec c]$ be a parsimonious affine reduction where $P \in \mathbb{Z}^{l \times k}, \vec c \in \mathbb{R}^{l}$, and the parameters of $D_2$ are non-trivial. Then the coefficients of $P$ are either 1, $-1$, or 0.
\end{lemma}
\begin{proof}
First suppose $\alpha_j$ in $D_1$ is trivial. In that case the non-trivial parameters of $D_2$ cannot depend on it (as the parsimonious reduction would make the parameters trivial themselves), so the coefficients of the $j$-th column of $P$ must be $0$.
Now suppose $\alpha_j$ is non-trivial. Then note that it must be mapped by the reduction to something non-zero (otherwise varying $\alpha_j$ in $D_1$ would result in the same diagram at $D_2$, meaning it is trivial). By the parsimonious property of $P$, it is mapped to a unique parameter $\beta_i$ in $D_2$. Write $\vec \alpha_j$ for the indicator vector $(\vec \alpha_j)_m = \delta_{jm} \alpha_j$ where we set every other parameter to zero, and define $\vec \beta_i$ similarly. Define the one-parameter diagrams $E_1[\alpha_j] := D_1[\vec \alpha_j]$ and $E_2[\beta_i] := D_2[\vec \beta_i + \vec c]$. Note that $\alpha_j$ and $\beta_i$ are non-trivial parameters in $E_1$ and $E_2$. 
Now let $r = P_{ij}$ be the coefficient corresponding to the mapping of $\alpha_j$ to $\beta_i$. Then
$$\input{linear-argument-2.tikz}$$
If $r \notin \{ 1, -1, 0 \}$ then note that 
$$E_1[2\pi/r] = \lambda' E_2[k 2\pi/r] = \lambda' E_2[2\pi] = \lambda' E_2[0] = \lambda' E_1[0]$$
for some constant $\lambda'$.
Hence, $E_1[2\pi/k] \propto E_1[0]$, so that $\alpha_j$ must be trivial in $E_1$ and hence in $D_1$, which is a contradiction. 
We conclude that the coefficients of the $j$-th column are either 1, $-1$, or 0.
\end{proof}

We can use this lemma to prove three useful consequences. First, any parsimonious reduction leads to a `phase teleportation' reduction, where the parameters are concentrated onto a smaller set of parameters in the \emph{original} diagram (Proposition~\ref{prop:phase-teleportation}). Second, any minimal reduction is also optimal (Corollary~\ref{cor:minimal-optimal}), and third, any parsimonious reduction can be written down by a simple ZX-diagram (Lemma~\ref{lem:P-construction}).
\begin{proposition}\label{prop:phase-teleportation}
    Let $D_1$ and $D_2$ be parametrised diagrams with respectively $k$ and $l$ parameters, $l<k$ and the parameters of $D_2$ are non-trivial, and let $D_1[\vec\alpha] = \lambda(\vec \alpha) D_2[P\vec\alpha + \vec c]$ be a parsimonious affine reduction. 
    Then there is a parametrised diagram $\widetilde{D}_1$ with $l$ parameters, which is equal as a diagram to $D_1$, but where a set of $k-l$ parameters have been set to zero, such that $D_1$ parsimoniously affinely reduces to $\widetilde{D}_1$. Specifically, we get a parsimonious reduction $D_1[\vec\alpha] = \lambda'(\vec\alpha)D_1[P'\vec\alpha]$ where $P'$ has $l$ non-zero rows.
\end{proposition}
\begin{proof}
  By Lemma~\ref{lem:params-non-duplicate} we may assume $P$ only contains elements $1$, $-1$ and $0$. As the reduction is parsimonious, each column only contains a single $\pm 1$. Without loss of generality we will assume that each parameter of $D_2$ is involved in the reduction, so that $P$ does not contain a zero row. Furthermore, by reordering the parameters in $\vec \alpha$, we may without loss of generality assume that $P$ looks like $P = (E | P')$ where $E=\text{diag}(\pm 1,\ldots, \pm 1)$, i.e.~the $i$th parameter of $D_1$ is sent to the $i$th parameter of $D_2$, and $P'$ is some other matrix that determines where the last $k-l$ parameters of $D_1$ are sent to. 
  Note that $P\vec \alpha$ is a vector of length $l$. We can pad it with zeroes to treat it as a vector of length $k$, which we denote as $P\vec \alpha;\vec 0$. We claim that $D_1[\vec \alpha] \propto D_1[(EP)\vec \alpha;\vec 0]$. To see this, note that
  \begin{align*}
    D_1[(EP)\vec\alpha;\vec 0] \ &=\ \lambda(EP\vec\alpha;\vec 0) D_2[P(EP\vec\alpha;\vec 0) + \vec c] \ = \ \lambda(EP\vec\alpha;\vec 0)  D_2[EEP\vec\alpha + P'\vec 0 + \vec c] \\
    &= \ \lambda(EP\vec\alpha;\vec 0) D_2[P\vec\alpha + \vec c] \ = \ \lambda(EP\vec\alpha;\vec 0)\lambda(\vec\alpha)^{-1} D_1[\vec\alpha].
  \end{align*}
  Now, let $\widetilde{D}_1$ be the diagram equal to $D_1$, but where the last $k-l$ parameters are set to zero, so that $\widetilde{D}_1[\vec \gamma] = D_1[\vec \gamma; \vec 0]$, then it is clear that there is a parsimonious reduction from $D_1$ to $\widetilde{D}_1$ given by $D_1[\vec \alpha] = \lambda'(\vec \alpha)\widetilde{D}_1[EP\vec\alpha]$ where $\lambda'(\vec\alpha)=\lambda(EP\vec\alpha;\vec 0)\lambda(\vec\alpha)^{-1}$.
\end{proof}

\begin{corollary}\label{cor:minimal-optimal}
  A parsimonious reduction that is minimal must also be optimal. Formally, let $D_1$, $D_2$ and $D_3$ be parametrised diagrams with only non-trivial parameters. Suppose that $D_1$ parsimoniously affinely reduces to $D_2$ and to $D_3$, and furthermore suppose that $D_2$ does not parsimoniously reduce to any smaller set of parameters (so $D_2$ is \emph{minimal}). Then $D_3$ has at least as many parameters as $D_2$ (meaning $D_2$ is \emph{optimal} among all parsimonious affine reductions of $D_1$).
\end{corollary}
\begin{proof}
  Let $k_1$, $k_2$ and $k_3$ be the number of parameters of $D_1$, $D_2$ and $D_3$. We can without loss of generality take $k_2\leq k_1$ and $k_3\leq k_1$, and we may assume that each parameter of $D_2$ is involved in the reduction from $D_1$ to $D_2$ by minimality. We need to show that $k_2\leq k_3$. 
  We will do this by constructing a new parsimonious reduction from $D_2$ to itself which has no more parameters than $\min(k_2, k_3)$. As $D_2$ is minimal by assumption, this implies that $k_2\leq k_3$ as required.

  Let $D_1[\vec\alpha] = \lambda(\vec\alpha)D_2[P_2\vec\alpha + \vec c]$ be the minimal parsimonious reduction.
  We can pick a right-inverse $Q$ of $P_2$ so that $P_2Q = \mathbb{I}$, because $P_2$ is parsimonious, and every parameter in $D_2$ has at least 1 preimage in $D_1$ by minimality (essentially, for each parameter $j$ in $D_2$, find a parameter $i$ in $D_1$ that maps to it, and then construct $Q$ that does this mapping $j\mapsto i$ including the right minus signs). Note that $Q$ is itself also parsimonious, with each column corresponding to a parameter of $D_2$ being sent to some parameter in $D_1$. Then note that 
  $$D_1[Q\vec\beta - Q\vec c] \ =\  \lambda(Q\vec\beta) D_2[P_2Q\vec\beta - P_2Q\vec c + \vec c] \ = \ \lambda(Q\vec\beta) D_2[\vec\beta].$$
  Hence, we have a parsimonious reduction from $D_2$ to $D_1$ given by $D_2[\vec\beta] = \lambda(Q\vec\beta)^{-1} D_1[Q\vec\beta - Q\vec c]$.
  By Proposition~\ref{prop:phase-teleportation}, the reduction from $D_1$ to $D_3$ implies a reduction from $D_1$ to itself: $D_1[\vec\alpha]=\lambda_1(\vec\alpha) D_1[P_1\vec\alpha]$ where $P_1$ has $k_3$ non-zero rows, so $\mathrm{rank}(P_1) = k_3$. Note that the converse is true for parsimonious reductions: if $\mathrm{rank}(P) = k$ then $P$ has exactly $k$ non-zero rows, as the columns of $P$ have at most 1 non-zero entry.

  Now we can compose these parsimonious reductions as follows:
  \begin{align*}
    D_2[\vec\beta] \ &=\  \lambda(Q\vec\beta)^{-1} D_1[Q\vec\beta - Q\vec c] \ \\
    &= \ \lambda_1(Q\vec\beta-Q\vec c)\lambda(Q\vec\beta)^{-1} D_1[P_1Q\vec\beta - P_1Q\vec c] \ \\
    &= \ \lambda(P_1Q\vec\beta - P_1Q\vec c)\lambda_1(Q\vec\beta-Q\vec c)\lambda(Q\vec\beta)^{-1} D_2[P_2P_1Q\vec\beta - P_2P_1Q\vec c + \vec c] \\
    &= \ \lambda'(\vec \beta) D_2[P'\vec\beta + \vec d]
  \end{align*}
  where here in the last step we defined $\lambda'(\vec \beta) = \lambda(P_1Q\vec\beta - P_1Q\vec c)\lambda_1(Q\vec\beta-Q\vec c)\lambda(Q\vec\beta)^{-1}$, $P' = P_2P_1Q$ and $\vec d = \vec c - P_2P_1Q\vec c$.
  $P'$ corresponds to a parsimonious reduction, as parsimonious reductions are closed under composition. This can be seen by considering the action of the overall reduction on individual parameters and checking that the parsimonious property is preserved.
  By minimality of $D_2$ we know that $\mathrm{rank}(P') = k_2$.
  But we can also calculate: $\mathrm{rank}(P') = \mathrm{rank}(P_2P_1Q) \leq \min(\mathrm{rank}(P_2), \mathrm{rank}(P_1), \mathrm{rank}(Q)) = \min(k_2, k_3)$.  Hence, $k_2 \leq k_3$ as required.
\end{proof}

We can also use Lemma~\ref{lem:params-non-duplicate} to show that any parsimonious reduction is particularly simple, meaning that we can implement the mapping of the parameters in the reduction using a small ZX-diagram.

\begin{lemma}\label{lem:P-construction}
  Let $D_1$ and $D_2$ be parametrised diagrams with respectively $k$ and $l$ parameters and where all the parameters of $D_1$ and $D_2$ are non-trivial. Let $D_1[\vec\alpha] = \lambda(\vec \alpha) D_2[P\vec\alpha + \vec c]$ be a parsimonious affine reduction and suppose each parameter of $D_2$ is involved in the reduction. Then there is a linear map $\hat P \in \mathbb{C}^{2^l \times 2^k}$ constructed using \input{p-generators.tikz} where $s_i\in\{0,1\}$ and $c_i\in\R$ such that
  $\hat P\bigotimes^{k}_{i=1} \ket{+_{\alpha_i}} \propto \bigotimes^{l}_{j=1} \ket{+_{\beta_j}}$ where $\vec \beta = P\vec\alpha + \vec c$. In addition $D'_1 = \lambda' D'_2 \circ \hat P$ for some constant scalar $\lambda'$.
  % Given two equivalent Clifford + Phase ZX diagrams $D_1[\vec\alpha] = D_2[\vec\beta]$ where $\beta = P\vec\alpha + \vec c$, $P \in \mathbb{R}^{m \times n}$ and $\vec c \in \mathbb{R}^{m}$, there is .
\end{lemma}
\begin{proof}
  Note first that the reduction being parsimonious, combined with all parameters of $D_2$ being involved in the reduction implies that $l\leq k$.
  Write $\vec \beta = P\vec\alpha + \vec c$.
  By Lemma~\ref{lem:params-non-duplicate}, $P$ only contains elements from $\{1,0,-1\}$, and there is at most one non-zero element in every column. Because all the parameters of $D_1$ are non-trivial, a column cannot be all zero, so that each column contains exactly 1 non-zero element, either $1$ or $-1$. Let $s_i$ be 1 if the non-zero element in the $i$th row is $-1$, and otherwise let $s_1$ be zero. Then for each $j$ we have $\beta_j = (-1)^{s_{k_{j1}}}\alpha_{k_{j1}} + \cdots + (-1)^{s_{k_{jn_j}}}\alpha_{k_{jn_j}} + c_j$ where $n_j$ is the number of non-zero elements in the $j$th row of $P$ and the $k_{j1}, \ldots, k_{jn_j}$ enumerate those positions. Hence:
  \begin{equation*}
    \input{P-map-construction.tikz}
  \end{equation*}
  Here $\lambda'(\vec \alpha) = \prod_j e^{is_j\alpha_{j}}$ is the phase that is needed for the rewrite $X\ket{+_\alpha} = e^{i\alpha} \ket{+_{-\alpha}}$. Note that for each $i$ there is a unique $k_{jm}$ that is equal to $i$, so that $\alpha_i$ occurs uniquely in the equation above. Hence, by using some swaps we can write:
  \begin{equation*}
    \input{P-map-construction2.tikz}
  \end{equation*}
  Here $\hat P$ is defined as the combination of elements from \input{p-generators.tikz} that makes this equation work.
  We now have 
  $$D_1[\vec \alpha] \ =\  \lambda(\vec \alpha) D_2[\beta] \ =\  \lambda(\vec \alpha)\lambda'(\vec \alpha) (D_2\circ \hat P)[\vec \alpha].$$
  We see then that Proposition~\ref{prop:scalar-is-constant} applies, and that hence $\lambda(\alpha)\lambda'(\alpha)$ must be some constant scalar $\lambda'$, and that $D'_1 = \lambda' D'_2\circ \hat P$.
\end{proof}

We are now almost ready to prove the main result. We just need the following lemma about stabilisers of weight 2 of graph states with local Cliffords.

\begin{lemma}\label{lem:zz-redex}
  Let $|\psi\> = (C_1 \otimes \ldots \otimes C_n)|G\>$ be a state expressed as a graph state $|G\>$ for graph $G$ followed by local Cliffords $C_i$. Assume that for a pair of distinct qubits $u, v \in \{ 1, \ldots, n \}$, $Z_u Z_v |\psi\> = |\psi\>$ and $C_u, C_v \in \{ I, H \}$. It must be the case that either:
  \begin{enumerate}
    \item[(i)] $u$ and $v$ are connected in $G$, have no other neighbours, and $C_u \otimes C_v$ is $I \otimes H$ or $H \otimes I$.
    \item[(ii)] $u$ and $v$ are not connected in $G$, have identical neighbours, and $C_u \otimes C_v = H\otimes H$.
  \end{enumerate}
\end{lemma}
\begin{proof}
  The layer of local Cliffords $C_1 \otimes \ldots \otimes C_n$ preserves the weight of stabilisers, so any stabiliser of the form $Z_uZ_v$ of $|\psi\>$ must come from a weight-2 stabiliser $P_u P_v$ of $|G\>$.
  The latter are generated by stabilisers of the form $S(w) := X_w \prod_{w'\in \textrm{nhd}(w)} Z_{w'}$ for all vertices $w$ in $G$. Any product of $k$ distinct generators has weight $\geq k$, hence a weight-2 stabiliser $S$ for $|G\>$ must be the product of 1 or 2 generators. If $S$ is a generator itself, it must be of the form $X_u Z_v$ or $Z_u X_v$, either of which will be transformed into $Z_u Z_v$ by local Cliffords when $C_u \otimes C_v = H \otimes I$ or $I \otimes H$, respectively. Hence $u$ and $v$ satisfy condition (i).

  If $S$ is the product of 2 distinct generators, then $S = S(u) S(v)$ and $\textrm{nhd}(u)\setminus\{v\} = \textrm{nhd}(v)\setminus\{u\}$, otherwise $S$ would have support at some qubit $\notin \{ u, v \}$. If $u$ and $v$ are connected, then $S = Y_u Y_v$, which cannot be translated into $Z_u Z_v$ via local Cliffords in the set $\{ I, H \}$. Therefore $u$ and $v$ must not be connected, $\textrm{nhd}(u) = \textrm{nhd}(v)$, and $S = X_u X_v$. This gives a stabiliser $Z_u Z_v$ for $|\psi\>$ when $C_u \otimes C_v = H\otimes H$. Hence $u$ and $v$ satisfy condition (ii).
\end{proof}

We can now prove the main theorem.

\begin{theorem}\label{thm:main-result}
% Given a parameterised ZX diagram $D_1[\vec\alpha]$, the PyZX simplification procedure produces a diagram $D_2[\vec\beta]$ such that $\vec\beta = P\vec\alpha$ and the length of $\vec\beta$ is minimal.
The parameter optimisation algorithm of~\cite{kissinger2019tcount} described in Section~\ref{sec:opt-parameters} is optimal. That is, if $C_1$ is a parametrised circuit, $C_2$ is the optimised circuit it reduces to found by~\cite{kissinger2019tcount}, and $C_3$ is any other parametrised circuit that $C_1$ parsimoniously affinely reduces to, then $C_3$ has at least as many parameters as $C_2$.
% Given a Clifford + Phase ZX diagram $D[\vec\alpha]$, the \texttt{full\_reduce} procedure produces a diagram $D_1[\vec\beta]$ such that the length of $\vec\beta$ is minimal, under the following assumptions of subsection \ref*{the-assumptions}.
\end{theorem}
\begin{proof}
Suppose we start with some parametrised Clifford circuit $C[\vec \alpha']$. Then all its parameters are non-trivial by Proposition~\ref{prop:circ-nontrivial}. Let $D_1[\vec\alpha]$ be the parametrised Clifford diagram produced by applying the ZX-based simplification strategy of Section~\ref{sec:opt-parameters} to $C$. Hence, $C$ parsimoniously affinely reduces to $D_1$. We need to show that $D_1$ is minimal. By Corollary~\ref{cor:minimal-optimal}, $D_1$ is then also optimal.

Suppose now that there exists a parametrised diagram $D_2[\vec\beta]$ such that $D_1[\vec\alpha] = \lambda(\vec \alpha) D_2[P\vec\alpha + \vec c]$ and that the length of $\vec\beta$ is less than that of $\vec\alpha$. We may assume without loss of generality that each parameter in $D_2$ is involved in the reduction (i.e.~that some parameter of $D_1$ is mapped to it), since otherwise we could just consider a $D_3$ equal to $D_2$ but with those parameters set to zero. Now, because all the parameters of $D_1$ are non-trivial and $P$ is parsimonious, $P$ must map the parameters only to parameters of $D_2$ that act non-trivially  in the subspace of the image. We may hence also without loss of generality assume that all the parameters of $D_2$ are non-trivial. By Lemma~\ref{lem:P-construction} we then have:
\ctikzfig{optimal-factor}
Here $\lambda'$ is some constant and $P$ is built from the generators {\input{p-generators-2.tikz} \begin{tikzpicture}
	\begin{pgfonlayer}{nodelayer}
		\node [style=white phase dot] (0) at (0, 0) {$c_i$};
		\node [style=none] (1) at (-1, 0) {};
		\node [style=none] (2) at (1, 0) {};
	\end{pgfonlayer}
	\begin{pgfonlayer}{edgelayer}
		\draw (1.center) to (0);
		\draw (0) to (2.center);
	\end{pgfonlayer}
\end{tikzpicture}
.} Since this equality holds for all $\vec \alpha$ we have:
\begin{equation}\label{eq:optimal-factor-state}
  \input{optimal-factor-state.tikz}
\end{equation}

Because $D_2[\vec\beta]$ has fewer parameters than $D_1[\vec\alpha]$, $\hat P$ should contain at least one $Z$ fusion \input{Z-fusion.tikz}. Let $\alpha_i$ and $\alpha_j$ be two parameters that fuse. Then define the projector $\Pi$ by the projection operation along these parameters:
$$\Pi \ =\  \input{optimal-stabiliser.tikz}$$
Note that we then have $\hat P \circ \Pi = \hat P$, which we can show using some spider fusions:
\[
  \input{stabiliser-fuse.tikz}
\]
Hence, using Eq.~\eqref{eq:optimal-factor-state} we see that $D'_1 \circ \Pi = \lambda' D'_2\circ \hat P \circ \Pi = \lambda' D'_2\circ \hat P = D'_1$, so that $\Pi$ stabilises $D'_1$.
We can see that $\Pi$ stabilises the Pauli $(-1)^{s_i+s_j} Z_iZ_j$, and hence this Pauli is also a stabiliser of $D'_1$. For simplicity, lets first assume $s_i+s_j = 0$ modulo 2, so that $Z_iZ_j$ is a stabiliser of $D'_1$.
Since $D'_1$ arose from the simplification algorithm, we know it is a GSLC diagram. Furthermore, the indices $i$ and $j$ correspond to where a parameter is plugged in, which means that the local Clifford there is either an identity or a Hadamard.
Now, applying Lemma~\ref{lem:zz-redex} we conclude that the vertices corresponding to the indices $i$ and $j$ are either connected to each other and not to anything else, or that they are not connected to each other and have an identical set of neighbours, with both their local Cliffords being a Hadamard. In the first case, the parameters only contribute a scalar factor, since they are not connected to the rest of the diagram, and hence they are trivial, which is a contradiction. The second case is however also not possible, since then the indices $i$ and $j$ would correspond to phase gadgets with identical neighbourhoods, in which case the algorithm would have already fused them.

We are hence led to a contradiction, which means that $D_1$ does not reduce to a diagram $D_2$ that has fewer parameters. Hence, by Corollary~\ref{cor:minimal-optimal}, the parameter count of $D_1$ is the optimal parameter count of the original circuit $C_1$.

If instead we had $s_i+s_j = 1$, then we can consider $D'_1$ precomposed with the Pauli $X_i^{s_i+s_j}$. This diagram then has $Z_iZ_j$ as a stabiliser. The Pauli $X$ can be absorbed into the diagram in a simple way, recovering its GSLC structure, so that we can again appeal to Lemma~\ref{lem:zz-redex} in the same way as before to conclude the indices $i$ and $j$ would correspond to phase gadgets with identical neighbourhoods.
\end{proof}

\begin{remark}\label{rem:post-selection}
  In this proof we only needed the fact that $C$ is a circuit to argue that its parameters are non-trivial, so that this optimality result applies to any type of diagram that has non-trivial parameters. In addition, we didn't make any assumptions on the hypothetical diagram $D_2$ other than that it has fewer parameters. This means that even allowing, for instance, post-selection in your circuits does not allow one to reduce the number of parameters further.
\end{remark}

\subsection{General affine reductions}\label{sec:general-affine}

Our optimality proof of Theorem~\ref{thm:main-result} requires the affine reductions to be parsimonious, i.e.~that each parameter only gets used in one place in the reduction. We expect however that a version of Theorem~\ref{thm:main-result} should continue to hold even for arbitrary affine reductions that are not parsimonious, because we suspect that each non-parsimonious reduction should either be pathological in some sense (such as involving trivial parameters), or either should be rewritable to a parsimonious reduction.

While we have not managed to prove this conjecture, in this section we will demonstrate some evidence towards this result. In particular, we will prove the following proposition.
\begin{proposition}
  Let $D_1[\vec\alpha] = D_2[M\vec \alpha +\vec c]$ be an affine reduction where every column of $M$ contains at most \emph{three} non-zero elements, and suppose that all parameters $\vec \alpha$ are non-trivial. Then there exists a parsimonious reduction $D_1[\vec \alpha] = D_2[M'\vec \alpha + \vec c]$.
\end{proposition}
Hence, as a corollary we can strengthen our Theorem~\ref{thm:main-result} to say that we have optimality over all affine reductions that copy a parameter to at most three places.

Note that proving this proposition is quite technical, and hence most details will be postponed to Appendix~\ref{app:non-parsimonious}.

First, the following lemma demonstrates that for a non-parsimonious reduction, something `odd' seems to be going on when we have a copied parameter: the diagram is invariant under a continuous group of phases. 

\begin{lemma}\label{lem:no-cloning-gadget}
  Suppose that $D_1[\alpha] = D_2[z_1\alpha + c_1,\ldots, z_k\alpha + c_k]$ with $z_1,\ldots, z_k \in \Z$ integers and $c_1,\ldots, c_k\in \R$ some constants. Suppose furthermore without loss of generality that the first $z_1,\ldots, z_l$ are odd and $z_{l+1},\ldots, z_k$ are even. Then $\forall \beta\in\R$ we have:
  \ctikzfig{no-cloning-pf-gadget2}
\end{lemma}
\begin{proof}
  Define $\tilde{D}_2$ as $D_2$ with the constant phases $\vec c\in \R^k$ added to the wires corresponding to the parameters. I.e.~$\tilde{D}_2[\vec \beta] := D_2[\vec \beta+\vec c]$. Then we can write $D_1[\alpha] = \tilde{D}_2[z_1\alpha,\ldots, z_k\alpha]$. Hence, without loss of generality we may assume that $\vec c = 0$ (at the cost of making $D_2$ a Clifford diagram with potential non-Clifford phases on its parameter wires).

  Now let $\beta\in \R$. Then for $a\in \{0,1\}$ we have $D_1[\beta + a\pi] = D_2[z_1\beta + a\pi,\ldots, z_l\beta + a\pi, z_{l+1}\beta,\ldots , z_k\beta]$, where the $a\pi$ doesn't appear on the last parameters, since those $z_j$ are even by assumption.
  % Then we note that $D_1[a\pi] = D_2[a\pi,\ldots, a\pi, 0\ldots 0]$ for $a\in \{0,1\}$ where the $a\pi$ appears on the right-hand side in the first $l$ spots. 
  Now we can do the following rewrite:
  \begin{equation}
    \input{no-cloning-pf-gadget.tikz}
  \end{equation}
  So now on the left-hand side and on the right-hand side we have a single input of $\ket{+_{a\pi}}$. Since these form a basis we conclude that this must hold for any input, and hence we can leave that wire open. Applying a $Z(-\beta)$ to both sides then gives us the desired equation.
\end{proof}

Note furthermore that this shows that any non-parsimonious reduction is equivalent to a parsimonious one where we precompose the diagram with a small number of spiders.

\begin{proposition}\label{prop:to-non-parsimonious}
  Let $D_1[\vec \alpha] = D_2[M \vec \alpha + \vec c]$ for some integer matrix $M$ and constants $\vec c$. Then $D_1[\vec \alpha] = (D_2'\circ E)[\vec \alpha]$ where $E$ is constructed out of Z-spiders with phases from $\vec c$ and fan-out X-spiders.
\end{proposition}
\begin{proof}
  Do the same construction as in Lemma~\ref{lem:no-cloning-gadget} iteratively where we fix every parameter except one, to extract it into $E$ and set $\beta=0$ in all cases.
\end{proof}

This does not imply that non-parsimonious reductions can always be reduced to a parsimonious one with the same parameter count. One could for instance have the non-parsimonious reduction $D_1[\alpha,\beta,\gamma] = D_2[\alpha+\beta,\alpha+\gamma]$, so that $D_2$ only has two parameters. The construction above then gives:
\begin{equation}
  \input{parsimonious-reduction.tikz}
\end{equation}
Hence, $D_2$ composed with the $E$ from Proposition~\ref{prop:to-non-parsimonious} has three parameters instead.
So while the construction of Proposition~\ref{prop:to-non-parsimonious} gives us a parsimonious reduction between $D_1'$ and $D_2'\circ E$, the diagram $D_2'\circ E$ might have more parameters than the original $D_2$. In addition, if there are non-Clifford constants $\vec c$, then these get incorporated in $E$, so that the reduction is no longer to a Clifford diagram.

We suspect however that every non-parsimonious reduction is `pathological'. For instance, that the only parameters that can get copied are those that are trivial anyway, or otherwise that the reduction function can be simplified to a parsimonious one. For example, suppose we have:
\[
  \input{D2-simple-non-parsimonious.tikz}
\]
Here $D_2$ has three parameters and $D_1$ has a single parameter, and we see that $D_1[\alpha] = D_2[\alpha,2\alpha,-2\alpha]$ using the spider fusion rule. But we also have the parsimonious reduction $D_1[\alpha] = D_2[\alpha,0,0]$, and so the copying of the parameters is not `necessary' in this case. 
This particular case of going from a non-parsimonious reduction to a parsimonious actually follows directly from Lemma~\ref{lem:no-cloning-gadget}: the one case where we can get rid of the copying of the parameter, is when only one of the $z_j$ is odd, with the rest being even: the X-spider we introduce in Lemma~\ref{lem:no-cloning-gadget} is then an identity and can be removed to get a straightforward parsimonious reduction.

However, to prove Proposition~\ref{prop:to-non-parsimonious} we need to find stronger consequences
from Lemma \ref{lem:no-cloning-gadget}. This is quite technical, and the full proof of Proposition~\ref{prop:to-non-parsimonious} requires several case distinctions, and so we postpone it to Appendix~\ref{app:non-parsimonious}.

\section{Measurement-based quantum computing}\label{sec:MBQC}

We have so far only talked about parameters in the context of parametrised circuits, but they are also a useful concept in measurement-based quantum computing (MBQC). The most well-studied form of MBQC is the \emph{one-way model}~\cite{MBQC1,MBQC2}. In this model you start with a graph state and you do measurements in one of the three principal planes of the Bloch sphere. Subsequent measurement angles may depend on previous measurement outcomes, which is needed in order to correct undesirable measurement outcomes and to get a deterministic result of the overall computation. All the instructions needed to carry out this computation are captured in a \emph{measurement pattern}, which specifies the \emph{open graph} of the computation (the graph state with designated vertices corresponding to in- and outputs), the order in which each vertex is measured, and in which measurement plane and angle it is measured depending on the previously observed measurement outcomes.
Note that for a given open graph it might not be possible to pick a measurement order and set of corrections that results in a deterministic computation. In order to ensure this we need certain combinatorial properties of the underlying open graph that ensure there is a \emph{flow} for the errors to be pushed along.
In particular, a measurement pattern is \emph{uniformly}, \emph{strongly} and \emph{stepwise} deterministic if and only if the underlying graph state has a \emph{generalised flow} (gflow)~\cite{GFlow}.

A pattern is uniformly deterministic if it is deterministic for every choice of measurement angle for the chosen measurement planes. It is strongly deterministic if every measurement outcome happens with equal probability, and it is stepwise deterministic if after every single measurement the set of corrections to make it deterministic is known again.
A generalised flow then is a property that an open graph can have which has to do with the neighbourhoods of the vertices and choice of measurement plane.

The most general flow condition is known as \emph{Pauli flow}. In this setting, qubits are either measured as some Pauli, or they are measured in an arbitrary angle in a given plane of the Bloch sphere. The uniform determinism condition is then weakened to only apply to those qubits measured in the arbitrary angle. When represented as a ZX-diagram, such patterns correspond to parametrised Clifford diagrams~\cite{Simmons2021Measurement}. In~\cite{wetering-gflow} and~\cite{Simmons2021Measurement} it is shown how to rewrite measurement patterns with gflow or Pauli flow in order to remove certain measured qubits while preserving the flow conditions. On the level of ZX-diagrams these rewrites correspond precisely to the strategy of~\cite{kissinger2019tcount} that we have shown to be optimal when it comes to parameter optimisation.

We hence also have the following result.
\begin{theorem}
  Let $\mathcal{D}$ be a measurement pattern with gflow which has an equal number of input and output vertices (so that it implements a unitary), and let $\mathcal{D}'$ be the optimised measurement pattern produced by the algorithm of \cite{wetering-gflow}. Then $\mathcal{D}'$ has an optimal number of parametrised measurements amongst those patterns with gflow that $\mathcal{D}$ affinely parsimoniously reduces to.
\end{theorem}
\begin{proof}
  $\mathcal{D}$ can be represented as a parametrised Clifford diagram. It can be turned into a unitary circuit using the circuit extraction technique of~\cite{wetering-gflow}. Each parametrised measurement becomes a $Z[\alpha]$ phase gate in the resulting extracted circuit, and hence by Proposition~\ref{prop:circ-nontrivial}, all parameters are not trivial. Each measurement pattern with gflow that $\mathcal{D}$ can reduce to can also be represented as a parametrised ZX-diagram, and can also be extracted to a unitary circuit, and so must only have non-trivial parameters. Hence, the result of Theorem~\ref{thm:main-result} applies, and the parametrised ZX-diagram we get by applying the rewrite strategy summarised in Section~\ref{sec:opt-parameters} has the optimal number of parameters amongst the diagrams that $\mathcal{D}$ parsimoniously affinely reduces to. This ZX rewrite strategy consists of the same rewrites as those described for measurement patterns in~\cite{wetering-gflow}, and hence the measurement pattern~\cite{wetering-gflow} produces indeed has the optimal number of parameters.
\end{proof}

\section{Conclusion and discussion}

We showed that we can find the minimal number of parameters of a Clifford circuit with parametrised phase gates under the condition that each parameter occurs uniquely and that the notion of reduction is restricted to analytic maps that don't clone any parameter to multiple places. The algorithm that finds the minimal parameter count is the one described in earlier work of some of the authors~\cite{kissinger2019tcount}. Our result also applies to measurement patterns with gflow and shows that the technique of~\cite{wetering-gflow} produces a measurement pattern with the minimal number of non-Clifford measurements (under the condition that we treat these non-Clifford angles as `black box parameters').

Note that the ZX-calculus-based approach of~\cite{kissinger2019tcount} was primarily interested in optimising T-count, and as an approach it always seemed to match the T-count of the circuit-based approach based on fusing Pauli exponentials of~\cite{zhang2019optimizing}. Evidence that these phase folding methods should result in the same non-Clifford count was further given in~\cite{Simmons2021Measurement}. Hence, it seems likely that the approach of~\cite{zhang2019optimizing} should also give optimal parameter counts. That means that writing a circuit as a series of Pauli exponentials and merging those exponentials of the same Pauli when there are only commuting Paulis in between them is essentially the best possible rewrite strategy when our goal is minimising the number of non-Clifford components of the circuit and we can't use any specific knowledge on the angles of rotation involved (barring the use of discontinuous parameter changing rules like Eq.~\eqref{eq:generalised-euler}). This result was shown formally in~\cite{vandaele2024optimal}, which appeared after our initial preprint.

Our results demonstrate that any circuit optimisation strategy that wants to do better at removing non-Clifford phases needs to use specific knowledge of the phases involved. This is indeed the case in T-count optimisation wherein relations to Reed-Muller decoding~\cite{amy2016t} and symmetric 3-tensor factorisation~\cite{heyfron2018efficient} are used. In particular, in~\cite{amy2016t} it was shown that for diagonal circuits built out of CNOT gates and phase gates, there are only non-trivial identities when the phases involved are dyadic rational multiples of $\pi$. 

There are several possible generalisations of our result that could be considered. We have already shown that generalising the discrete gate set from Cliffords to Clifford+T results in a parameter optimisation problem that is NP-hard, but we left open the question of the hardness of other possible generalisations. In particular, we have only established the hardness of optimisation when parameters are allowed to be reused in the setting of post-selected quantum circuits. We don't know if the problem remains hard in the unitary setting. This is particularly relevant, because many ans\"{a}tze for variational circuits involve controlled-phase gates \cite{sim2019expressibility, havlivcek2019supervised} that get decomposed into repeated parametrised phase gates. We have also left open the hardness when we relax the notion of reduction between parametrised circuits to include parameter transformations that are discontinuous in parameter space, like the Euler angle transformation of Eq.~\eqref{eq:generalised-euler}. This rule suffices to make the ZX-calculus complete for arbitrary linear maps~\cite{vilmarteulercompleteness} so it stands to reason that rewriting when this rule is allowed becomes hard.
Lastly, our optimality result applies to \emph{parsimonious} reductions: those that do not clone the parameter. We've showed in Section~\ref{sec:general-affine} that when a parameter is used at most 3 times in a reduction, that we can also find an equivalent parsimonious reduction. We suspect that this should hold for any non-parsimonious reduction, and hence that our optimality result should hold even when allowing any affine reduction.

\medskip
\noindent\textbf{Acknowledgments}: 
This work is supported by the Engineering and Physical Sciences Research Council grant number EP/Z002230/1: (De)constructing quantum software (DeQS). RY would like to thank Simon Harrison for his generous support via the Wolfson Harrison UKRI Quantum Foundation Scholarship. JvdW is supported by the NWO Veni personal fellowship.

\bibliographystyle{plainnat}
\bibliography{bibliography}

\begin{thebibliography}{36}
\providecommand{\natexlab}[1]{#1}
\providecommand{\url}[1]{\texttt{#1}}
\expandafter\ifx\csname urlstyle\endcsname\relax
  \providecommand{\doi}[1]{doi: #1}\else
  \providecommand{\doi}{doi: \begingroup \urlstyle{rm}\Url}\fi

\bibitem[Amy et~al.(2013)Amy, Maslov, Mosca, and
  Roetteler]{meet-in-the-middle2013}
M.~Amy, D.~Maslov, M.~Mosca, and M.~Roetteler.
\newblock A meet-in-the-middle algorithm for fast synthesis of depth-optimal
  quantum circuits.
\newblock \emph{IEEE Transactions on Computer-Aided Design of Integrated
  Circuits and Systems}, 32\penalty0 (6):\penalty0 818--830, 6 2013.
\newblock ISSN 0278-0070.
\newblock \doi{10.1109/TCAD.2013.2244643}.

\bibitem[Amy(2019)]{AmyVerification}
Matthew Amy.
\newblock Towards large-scale functional verification of universal quantum
  circuits.
\newblock In Peter Selinger and Giulio Chiribella, editors, \emph{{Proceedings
  of the 15th International Conference on} Quantum Physics and Logic, {Halifax,
  Canada, 3-7th June 2018}}, volume 287 of \emph{Electronic Proceedings in
  Theoretical Computer Science}, pages 1--21. Open Publishing Association,
  2019.
\newblock \doi{10.4204/EPTCS.287.1}.

\bibitem[Amy and Mosca(2019)]{amy2016t}
Matthew Amy and Michele Mosca.
\newblock {T-count optimization and Reed-Muller codes}.
\newblock \emph{Transactions on Information Theory}, 2019.
\newblock \doi{10.1109/TIT.2019.2906374}.
\newblock URL \url{https://ieeexplore.ieee.org/document/8672175}.

\bibitem[Backens(2014)]{BackensCompleteness}
Miriam Backens.
\newblock The {ZX}-calculus is complete for stabilizer quantum mechanics.
\newblock \emph{New Journal of Physics}, 16\penalty0 (9):\penalty0 093021,
  2014.
\newblock \doi{10.1088/1367-2630/16/9/093021}.

\bibitem[Backens and Kissinger(2019)]{backens2018zhcalculus}
Miriam Backens and Aleks Kissinger.
\newblock {ZH}: A complete graphical calculus for quantum computations
  involving classical non-linearity.
\newblock In Peter Selinger and Giulio Chiribella, editors, \emph{Proceedings
  of the 15th International Conference on Quantum Physics and Logic, Halifax,
  Canada, 3-7th June 2018}, volume 287 of \emph{Electronic Proceedings in
  Theoretical Computer Science}, pages 18--34. Open Publishing Association,
  2019.
\newblock \doi{10.4204/EPTCS.287.2}.

\bibitem[Backens et~al.(2021)Backens, Miller-Bakewell, de~Felice, Lobski, and
  van~de Wetering]{wetering-gflow}
Miriam Backens, Hector Miller-Bakewell, Giovanni de~Felice, Leo Lobski, and
  John van~de Wetering.
\newblock {There and back again: A circuit extraction tale}.
\newblock \emph{{Quantum}}, 5:\penalty0 421, 3 2021.
\newblock ISSN 2521-327X.
\newblock \doi{10.22331/q-2021-03-25-421}.

\bibitem[Backens et~al.(2023)Backens, Kissinger, Miller-Bakewell, van~de
  Wetering, and Wolffs]{zhcompleteness2020}
Miriam Backens, Aleks Kissinger, Hector Miller-Bakewell, John van~de Wetering,
  and Sal Wolffs.
\newblock {Completeness of the ZH-calculus}.
\newblock \emph{{Compositionality}}, 5, 7 2023.
\newblock ISSN 2631-4444.
\newblock \doi{10.32408/compositionality-5-5}.

\bibitem[Browne et~al.(2007)Browne, Kashefi, Mhalla, and Perdrix]{GFlow}
Daniel~E. Browne, Elham Kashefi, Mehdi Mhalla, and Simon Perdrix.
\newblock Generalized flow and determinism in measurement-based quantum
  computation.
\newblock \emph{New Journal of Physics}, 9\penalty0 (8):\penalty0 250, 2007.
\newblock \doi{10.1088/1367-2630/9/8/250}.

\bibitem[Coecke and Duncan(2008)]{CD1}
Bob Coecke and Ross Duncan.
\newblock Interacting quantum observables.
\newblock In \emph{Proceedings of the 37th International Colloquium on
  Automata, Languages and Programming (ICALP)}, Lecture Notes in Computer
  Science, 2008.
\newblock \doi{10.1007/978-3-540-70583-3_25}.

\bibitem[Coecke and Duncan(2011)]{CD2}
Bob Coecke and Ross Duncan.
\newblock Interacting quantum observables: categorical algebra and
  diagrammatics.
\newblock \emph{New Journal of Physics}, 13:\penalty0 043016, 2011.
\newblock \doi{10.1088/1367-2630/13/4/043016}.

\bibitem[Duncan and Perdrix(2010)]{DP2}
Ross Duncan and Simon Perdrix.
\newblock {Rewriting Measurement-Based Quantum Computations with Generalised
  Flow}.
\newblock In \emph{Proceedings of {ICALP}}, Lecture Notes in Computer Science,
  pages 285--296. Springer, 2010.
\newblock \doi{10.1007/978-3-642-14162-1_24}.

\bibitem[Duncan et~al.(2020)Duncan, Kissinger, Perdrix, and van~de
  Wetering]{cliffsimp}
Ross Duncan, Aleks Kissinger, Simon Perdrix, and John van~de Wetering.
\newblock {Graph-theoretic Simplification of Quantum Circuits with the
  ZX-calculus}.
\newblock \emph{{Quantum}}, 4:\penalty0 279, 6 2020.
\newblock ISSN 2521-327X.
\newblock \doi{10.22331/q-2020-06-04-279}.

\bibitem[Farhi et~al.(2014)Farhi, Goldstone, and Gutmann]{farhi2014quantum}
Edward Farhi, Jeffrey Goldstone, and Sam Gutmann.
\newblock A quantum approximate optimization algorithm.
\newblock \emph{arXiv preprint arXiv:1411.4028}, 2014.
\newblock URL \url{https://arxiv.org/abs/1411.4028}.

\bibitem[Havl{\'\i}{\v{c}}ek et~al.(2019)Havl{\'\i}{\v{c}}ek, C{\'o}rcoles,
  Temme, Harrow, Kandala, Chow, and Gambetta]{havlivcek2019supervised}
Vojt{\v{e}}ch Havl{\'\i}{\v{c}}ek, Antonio~D C{\'o}rcoles, Kristan Temme,
  Aram~W Harrow, Abhinav Kandala, Jerry~M Chow, and Jay~M Gambetta.
\newblock Supervised learning with quantum-enhanced feature spaces.
\newblock \emph{Nature}, 567\penalty0 (7747):\penalty0 209--212, 2019.
\newblock \doi{10.1038/s41586-019-0980-2}.

\bibitem[Heyfron and Campbell(2018)]{heyfron2018efficient}
Luke~E Heyfron and Earl~T Campbell.
\newblock {An efficient quantum compiler that reduces T count}.
\newblock \emph{Quantum Science and Technology}, 4\penalty0 (015004), 2018.
\newblock \doi{10.1088/2058-9565/aad604}.

\bibitem[Jeandel et~al.(2018{\natexlab{a}})Jeandel, Perdrix, and
  Vilmart]{JPV-universal}
Emmanuel Jeandel, Simon Perdrix, and Renaud Vilmart.
\newblock {Diagrammatic Reasoning Beyond Clifford+T Quantum Mechanics}.
\newblock In \emph{Proceedings of the 33rd Annual ACM/IEEE Symposium on Logic
  in Computer Science}, LICS '18, pages 569--578, New York, NY, USA,
  2018{\natexlab{a}}. ACM.
\newblock ISBN 978-1-4503-5583-4.
\newblock \doi{10.1145/3209108.3209139}.

\bibitem[Jeandel et~al.(2018{\natexlab{b}})Jeandel, Perdrix, and
  Vilmart]{SimonCompleteness}
Emmanuel Jeandel, Simon Perdrix, and Renaud Vilmart.
\newblock {A Complete Axiomatisation of the ZX-Calculus for Clifford+T Quantum
  Mechanics}.
\newblock In \emph{Proceedings of the 33rd Annual ACM/IEEE Symposium on Logic
  in Computer Science}, LICS '18, pages 559--568, New York, NY, USA,
  2018{\natexlab{b}}. ACM.
\newblock ISBN 978-1-4503-5583-4.
\newblock \doi{10.1145/3209108.3209131}.

\bibitem[Kissinger and van~de Wetering(2020)]{kissinger2019tcount}
Aleks Kissinger and John van~de Wetering.
\newblock {Reducing the number of non-Clifford gates in quantum circuits}.
\newblock \emph{Physical Review A}, 102:\penalty0 022406, 8 2020.
\newblock \doi{10.1103/PhysRevA.102.022406}.

\bibitem[Kuijpers et~al.(2019)Kuijpers, van~de Wetering, and
  Kissinger]{GraphicalFourier2019}
Stach Kuijpers, John van~de Wetering, and Aleks Kissinger.
\newblock Graphical {F}ourier theory and the cost of quantum addition.
\newblock \emph{arXiv preprint arXiv:1904.07551}, 2019.
\newblock URL \url{https://arxiv.org/abs/1904.07551}.

\bibitem[McElvanney and Backens(2023)]{mcelvanney2022complete}
Tommy McElvanney and Miriam Backens.
\newblock Complete flow-preserving rewrite rules for {MBQC} patterns with
  {P}auli measurements.
\newblock In Stefano Gogioso and Matty Hoban, editors, \emph{Proceedings 19th
  International Conference on Quantum Physics and Logic, Wolfson College,
  Oxford, UK, 27 June - 1 July 2022}, volume 394 of \emph{Electronic
  Proceedings in Theoretical Computer Science}, pages 66--82. Open Publishing
  Association, 2023.
\newblock \doi{10.4204/EPTCS.394.5}.

\bibitem[Miller-Bakewell(2020)]{MillerBakewell2020finite}
Hector Miller-Bakewell.
\newblock {Finite Verification of Infinite Families of Diagram Equations}.
\newblock In Bob Coecke and Matthew Leifer, editors, \emph{Proceedings 16th
  International Conference on Quantum Physics and Logic, Chapman University,
  Orange, CA, USA., 10-14 June 2019}, volume 318 of \emph{Electronic
  Proceedings in Theoretical Computer Science}, pages 27--52. Open Publishing
  Association, 2020.
\newblock \doi{10.4204/EPTCS.318.3}.

\bibitem[Nam et~al.(2018)Nam, Ross, Su, Childs, and Maslov]{nam2018automated}
Yunseong Nam, Neil~J Ross, Yuan Su, Andrew~M Childs, and Dmitri Maslov.
\newblock Automated optimization of large quantum circuits with continuous
  parameters.
\newblock \emph{npj Quantum Information}, 4\penalty0 (1):\penalty0 23, 2018.
\newblock \doi{10.1038/s41534-018-0072-4}.

\bibitem[Perdrix and Wang(2016)]{supplementarity}
Simon Perdrix and Quanlong Wang.
\newblock Supplementarity is necessary for quantum diagram reasoning.
\newblock In \emph{41st International Symposium on Mathematical Foundations of
  Computer Science (MFCS 2016)}, volume~58 of \emph{Leibniz International
  Proceedings in Informatics (LIPIcs)}, pages 76:1--76:14, Krakow, Poland,
  2016.
\newblock \doi{10.4230/LIPIcs.MFCS.2016.76}.

\bibitem[Peruzzo et~al.(2014)Peruzzo, McClean, Shadbolt, Yung, Zhou, Love,
  Aspuru-Guzik, and O’brien]{peruzzo2014variational}
Alberto Peruzzo, Jarrod McClean, Peter Shadbolt, Man-Hong Yung, Xiao-Qi Zhou,
  Peter~J Love, Al{\'a}n Aspuru-Guzik, and Jeremy~L O’brien.
\newblock A variational eigenvalue solver on a photonic quantum processor.
\newblock \emph{Nature communications}, 5:\penalty0 4213, 2014.
\newblock \doi{10.1038/ncomms5213}.

\bibitem[Poór et~al.(2023)Poór, Booth, Carette, van~de Wetering, and
  Yeh]{poor2023qupit}
Boldizsár Poór, Robert~I. Booth, Titouan Carette, John van~de Wetering, and
  Lia Yeh.
\newblock {The Qupit Stabiliser ZX-travaganza: Simplified Axioms, Normal Forms
  and Graph-Theoretic Simplification}.
\newblock In Shane Mansfield, Benoit Val\^iron, and Vladimir Zamdzhiev,
  editors, \emph{{\rm Proceedings of the Twentieth International Conference on}
  Quantum Physics and Logic, {\rm Paris, France, 17-21st July 2023}}, volume
  384 of \emph{Electronic Proceedings in Theoretical Computer Science}, pages
  220--264. Open Publishing Association, 2023.
\newblock \doi{10.4204/EPTCS.384.13}.

\bibitem[Raussendorf and Briegel(2001)]{MBQC1}
Robert Raussendorf and Hans~J. Briegel.
\newblock {A One-Way Quantum Computer}.
\newblock \emph{Physical Review Letters}, 86:\penalty0 5188--5191, 5 2001.
\newblock \doi{10.1103/PhysRevLett.86.5188}.

\bibitem[Raussendorf et~al.(2003)Raussendorf, Browne, and Briegel]{MBQC2}
Robert Raussendorf, Dan~E. Browne, and Hans~J. Briegel.
\newblock Measurement-based quantum computation on cluster states.
\newblock \emph{Physical Review A}, 68\penalty0 (2):\penalty0 22312, 2003.
\newblock ISSN 1094-1622.
\newblock \doi{10.1103/PhysRevA.68.022312}.

\bibitem[Schuld et~al.(2019)Schuld, Bergholm, Gogolin, Izaac, and
  Killoran]{schuld2019evaluating}
Maria Schuld, Ville Bergholm, Christian Gogolin, Josh Izaac, and Nathan
  Killoran.
\newblock Evaluating analytic gradients on quantum hardware.
\newblock \emph{Physical Review A}, 99\penalty0 (3):\penalty0 032331, 2019.
\newblock \doi{10.1103/PhysRevA.99.032331}.

\bibitem[Sim et~al.(2019)Sim, Johnson, and Aspuru-Guzik]{sim2019expressibility}
Sukin Sim, Peter~D Johnson, and Al{\'a}n Aspuru-Guzik.
\newblock Expressibility and entangling capability of parameterized quantum
  circuits for hybrid quantum-classical algorithms.
\newblock \emph{Advanced Quantum Technologies}, 2\penalty0 (12):\penalty0
  1900070, 2019.
\newblock \doi{10.1002/qute.201900070}.

\bibitem[Simmons(2021)]{Simmons2021Measurement}
Will Simmons.
\newblock {Relating Measurement Patterns to Circuits via Pauli Flow}.
\newblock In Chris Heunen and Miriam Backens, editors, \emph{Proceedings 18th
  International Conference on Quantum Physics and Logic, Gdansk, Poland, and
  online, 7-11 June 2021}, volume 343 of \emph{Electronic Proceedings in
  Theoretical Computer Science}, pages 50--101. Open Publishing Association,
  2021.
\newblock \doi{10.4204/EPTCS.343.4}.

\bibitem[van~de Wetering(2020)]{vandewetering2020zxcalculus}
John van~de Wetering.
\newblock {ZX-calculus for the working quantum computer scientist}.
\newblock \emph{arXiv preprint arXiv:2012.13966}, 2020.
\newblock URL \url{https://arxiv.org/abs/2012.13966}.

\bibitem[van~de Wetering and Amy(2023)]{wetering2023optimising}
John van~de Wetering and Matt Amy.
\newblock {Optimising quantum circuits is generally hard}.
\newblock \emph{arXiv preprint arXiv:2310.05958}, 2023.
\newblock URL \url{https://arxiv.org/abs/2310.05958}.

\bibitem[Van Den~Nest(2010)]{VandenNest2010Classical}
Maarten Van Den~Nest.
\newblock Classical simulation of quantum computation, the gottesman-knill
  theorem, and slightly beyond.
\newblock \emph{Quantum Info. Comput.}, 10\penalty0 (3):\penalty0 258–271,
  2010.
\newblock ISSN 1533-7146.
\newblock \doi{10.5555/2011350.2011356}.

\bibitem[Vandaele et~al.(2024)Vandaele, Perdrix, and
  Vuillot]{vandaele2024optimal}
Vivien Vandaele, Simon Perdrix, and Christophe Vuillot.
\newblock Optimal number of parametrized rotations and hadamard gates in
  parametrized clifford circuits with non-repeated parameters.
\newblock \emph{arXiv preprint arXiv:2407.07846}, 2024.
\newblock URL \url{https://arxiv.org/abs/2407.07846}.

\bibitem[Vilmart(2019)]{vilmarteulercompleteness}
Renaud Vilmart.
\newblock {A Near-Minimal Axiomatisation of ZX-Calculus for Pure Qubit Quantum
  Mechanics}.
\newblock In \emph{2019 34th Annual ACM/IEEE Symposium on Logic in Computer
  Science (LICS)}, pages 1--10, 2019.
\newblock \doi{10.1109/LICS.2019.8785765}.

\bibitem[Zhang and Chen(2019)]{zhang2019optimizing}
Fang Zhang and Jianxin Chen.
\newblock Optimizing {T} gates in {Clifford+T} circuit as $\pi/4$ rotations
  around {Paulis}.
\newblock \emph{arXiv preprint arXiv:1903.12456}, 2019.
\newblock URL \url{https://arxiv.org/abs/1903.12456}.

\end{thebibliography}

\appendix

\section{Hardness of optimisation with repeated parameters}\label{app:repeated-params}

We have not yet determined the hardness of optimisation of Clifford circuits with repeated parameters. 
The argument of Proposition~\ref{prop:hardness-clifford-T} does not work, since we cannot `reliably' produce a Toffoli gate using just Clifford gates and parametrised phase gates: if all the parameters are set to Clifford angles, then the whole circuit is Clifford and hence certainly cannot implement a Toffoli. It turns out that it is however possible to construct a Toffoli for any other choice of parameter using repeated parametrised phase gates if we allow the circuit to be \emph{post-selected}.

In order to define the reduction problem for post-selected quantum circuits we slightly relax Definition~\ref{def:reduction} to allow for a zero scalar.
\begin{definition}
  We say a parametrised post-selected circuit $C_1$ with parameter space $\R^k$ \emph{reduces} to post-selected circuit $C_2$ with parameter space $\R^l$ when there exists a function $f:\R^k\to \R^l$ such that $C_1[\vec \alpha] = \lambda(\vec \alpha) C_2[f(\vec \alpha)]$ for all $\vec \alpha\in \R^k$ where $\lambda: \R^k\to \CC$ is a function representing a global scalar $\lambda(\vec \alpha)$ that may depend on the parameters and is allowed to be zero. We define an affine reduction as before by restricting $f$ to be affine.
\end{definition}

\begin{definition}  \textbf{(Affine) Parameter optimisation for post-selected circuits}: Given a parametrised post-selected quantum circuit $C_1$ with parameter space $\R^k$, find a parametrised quantum circuit $C_2$ with parameter space $\R^l$ that it (affinely) reduces to, such that $l$ is \emph{minimal}.

\end{definition}

\noindent In this appendix we will prove the following result.
\begin{proposition}
  Boolean satisfiability reduces to affine parameter optimisation for post-selected circuits. 
\end{proposition}

In order to prove this it will be helpful to use the \emph{ZH-calculus}~\cite{backens2018zhcalculus,zhcompleteness2020}. In particular we define a new type of generator, the \emph{H-box}:
\begin{equation}
  \input{H-spider.tikz} \ = \ \sum_{\vec x, \vec y} a^{x_1\cdots x_m y_1\cdots y_n} \ketbra{\vec y}{\vec x}, \qquad a \in \mathbb{C}
\end{equation}
This is a matrix that has a $1$ in every entry except for the bottom-right corner (where $x_1=\cdots=x_m=y_1=\cdots=y_n=1$) where there is an $a$. In particular, the 1-input 1-output H-spider with $a=-1$ is a rescaled Hadamard. For this reason, when $a=-1$ we don't write the label in the H-box.

The reason we are interested in H-boxes, is because they allow us to easily represent the AND operation $\text{AND}\ket{x,y} = \ket{xy}$:
\begin{equation}
  \input{H-box-AND.tikz}
\end{equation}
Hence, we can use H-boxes to construct the Toffoli:
\begin{equation}
  \input{toffoli-hbox.tikz}
\end{equation}
In order to relate H-boxes to parametrised phase gates we need to decompose it into the generators of the ZX-calculus. How to do this is described in~\cite{GraphicalFourier2019}:
\begin{equation}\label{eq:Fourier}
  \begin{tikzpicture}
	\begin{pgfonlayer}{nodelayer}
		\node [style=hadamard] (0) at (0, 0) {$e^{i\alpha}$};
		\node [style=none] (2) at (1.5, 0) {};
	\end{pgfonlayer}
	\begin{pgfonlayer}{edgelayer}
		\draw (0) to (2.center);
	\end{pgfonlayer}
\end{tikzpicture}
 \ = \ \begin{tikzpicture}
	\begin{pgfonlayer}{nodelayer}
		\node [style=white phase dot] (0) at (0, 0) {$\alpha$};
		\node [style=none] (2) at (1.25, 0) {};
	\end{pgfonlayer}
	\begin{pgfonlayer}{edgelayer}
		\draw (0) to (2.center);
	\end{pgfonlayer}
\end{tikzpicture}
 \qquad\quad \input{had-gadget-alpha.tikz} \qquad\quad \input{H-box-3.tikz} \ \propto \ \input{H-box-3-fourier.tikz}
\end{equation}
To decompose an $n$-ary H-box with an $e^{i\alpha}$ phase we require $2^n-1$ phase gadgets with $\pm \alpha/2^{n-1}$ phases. Now to show how to implement a Toffoli using parametrised phase gates we will need to additional rewrites, \emph{H-box fusion} and \emph{supplementarity}:
\begin{equation}\label{eq:H-fusion}
  \input{H-spider-rule-phased.tikz} \qquad\quad \input{supplementarity.tikz}
\end{equation}
We then show the following:
\begin{equation}\label{eq:tuomas-proof}
  \input{tof-from-phases.tikz}
\end{equation}
From this we conclude that:
\begin{equation}\label{eq:AND-from-parameters}
  \input{AND-from-parameters.tikz}
\end{equation}
This equation holds whenever $\alpha \neq 0$ (since otherwise the scalar factor in Eq.~\eqref{eq:tuomas-proof} would be zero). Since we want to work with integer multiples of $\alpha$ we can replace $\alpha$ by $4\alpha$, and the condition $\alpha\neq 0$ with the condition $\alpha\neq k\frac\pi2$ in this equation (because the equation must hold regardless of the value of $\alpha$ this substitution is valid):
\begin{equation}
  \input{AND-from-parameters2.tikz} \quad \text{when }\alpha \neq k\frac\pi2
\end{equation}
In order to make the connection to post-selected quantum circuits more clear we can use this construction to build a Toffoli as a post-selected circuit using controlled-phase gates:
\begin{equation}
  \input{TOF-from-parameters.tikz}
\end{equation}
Since we can build a Toffoli (up to a scalar that depends on the parameters), we can construct for every Boolean function $f:\{0,1\}^n\to \{0,1\}$ the map $U_f[\alpha]$ acting as $U_f[\alpha]\ket{\vec x,y} = \lambda(\vec \alpha)\ket{\vec x,y\oplus f(\vec x)}$ for some scalar function $\lambda$ which goes to zero when $\alpha=k\frac\pi2$. Note that we only need the single parameter $\alpha$ as we can construct each Toffoli using the same repeated parameter. In particular, if the construction of $U_f$ involves $k$ Toffolis, then the scalar is $\lambda(\alpha) = (1-e^{i4\alpha})^k$. 
Now with $U_f$ in hand, we can build the same circuit as in Proposition~\ref{prop:hardness-clifford-T} using two additional parameters $\beta$ and $\gamma$ to get the diagonal map $\ket{\vec x, y} \mapsto \lambda'(\alpha) e^{i\beta y + i\gamma (f(\vec x)\oplus y)}\ket{\vec x,y}$. Let's call the post-selected circuit implementing this map $C_f[\alpha,\beta,\gamma]$. 

Note that if $f$ is not satisfiable that then $\beta$ and $\gamma$ can fuse, meaning $C_f$ reduces to a circuit $C'[\delta]$ implementing $C'[\delta]\ket{\vec x,y} = e^{i\delta y}\ket{\vec x,y}$. The reduction is $C_f[\alpha,\beta,\gamma] = \lambda'(\alpha)C'[\beta+\gamma]$. Hence, the optimal parameter count of $C_f$ is $1$ (because it can't be zero). If instead $f$ is always satisfiable, then we can also show that it reduces to a circuit with parameter count 1. When $f$ is satisfiable, but not always satisfiable then we can show that $\beta$ and $\gamma$ have separate effects on at least two different input states to $C_f$ so that the optimal parameter count is at least $2$.

We can now state our reduction from SAT to parameter optimisation with post-selected quantum circuits. For a given Boolean formula $f$, construct the post-selected parametrised circuit $C_f$. Then determine its optimal parameter count. If this is more than $1$, then $f$ is satisfiable. If it is exactly $1$, then it is either not satisfiable or always satisfiable. Test the value $f(0\cdots 0)$. If this is $1$ then $f$ is satisfiable, and if it is not, then it is not satisfiable everywhere and hence must not be satisfiable at all.

\section{Affine functions are all you need}\label{app:affine-functions}

We will here formalise what we mean by the parameter maps consisting of analytic functions on the unit circle.
We require a parameter map $\Phi = (\Phi_1,\ldots, \Phi_l)$ to consist of functions $\Phi_j(\vec \alpha) = 2\pi f((\vec \alpha \text{ mod } 2\pi)/(2\pi)) \text{ mod } 2\pi$ for some smooth $f$. 
Here we are dividing and multiplying by $2\pi$ to make the following arguments more convenient. 
We additionally take the modulo $2\pi$ inside of the argument of $f$ so that we can view $f$ as a function $f:[0,1]^n\to \R$. 
To make $\Phi_j$ a proper map of the unit circle we need $f(1,x_2,\ldots, x_n) - f(0,x_2,\ldots, x_n) = k$ for some integer $k\in\mathbb{Z}$ and arbitrary $x_j \in [0,1]$, and similarly for the other arguments (note that this $k$ cannot depend on the choice of $x_2,\ldots x_n$ due to continuity of $f$). In addition, for $f$ to be smooth as a function on the unit circle we also need $f'(0,x_2,\ldots,x_n) = f'(1,x_2,\ldots, x_n)$, so that there is no discontinuity in the derivative at the boundary. The same holds for all higher derivatives as well.
We assume that $f$ is in each argument \emph{analytic}, meaning that it is infinitely differentiable, and given by a power series.  Concretely, writing $g(x) := f(x,x_2,\ldots, x_n)$ for some fixed choice of $x_2,\ldots, x_n$, we will assume that
\begin{equation*}
  g(x) \ := \ \sum_{m=0}^\infty a_m (x-b)^m
\end{equation*}
for some choice of $b, a_m \in \R$, such that this infinite sum absolutely converges within an open interval $I$ containing  $[-1,1]$.

Our discussion above then shows that we should assume that $g(1) -g(0) = k\in \mathbb{Z}$, and $g^{(j)}(0) = g^{(j)}(1)$, where $g^{(j)}$ denotes the $j$th derivative of $g$.

% We are assuming that our parameter map $\Phi = (\Phi_1,\ldots, \Phi_l)$ is given by 
% $$\Phi_j(\vec \alpha) = 2\pi f_j((\vec \alpha \text{ mod } 2\pi)/(2\pi)) \text{ mod } 2\pi$$ 
% for some $f_j$ satisfying the conditions outlined earlier. Namely, $f_j:[0,1]^k\to [0,1]$ is smooth, and in each separate parameter analytic. In addition $f(1,x_2,\ldots, x_n) - f(0,x_2,\ldots, x_n) = k$ for some integer $k\in\mathbb{Z}$ and arbitrary $x_j \in [0,1]$, and similarly for the other arguments. Additionally all the derivatives of $f$ agree at $0$ and $1$: $f^{(l)}(0,x_2,\ldots,x_n) = f^{(l)}(1,x_2,\ldots, x_n)$.

% Concretely, writing $g(x) := f(x,x_2,\ldots, x_n)$ for some fixed choice of $x_2,\ldots, x_n$, we will assume that
% \begin{equation*}
%   g(x) \ := \ \sum_{m=0}^\infty a_m (x-b)^m
% \end{equation*}
% for some choice of $b, a_m \in \R$, such that this infinite sum absolutely converges within an open interval $I$ containing  $[-1,1]$.

% Our assumptions then mean that $g(1) -g(0) = k\in \mathbb{Z}$, and $g^{(j)}(0) = g^{(j)}(1)$, where $g^{(j)}$ denotes the $j$th derivative of $g$.

\begin{lemma}
  Let $g$ be an analytic function satisfying the conditions $g(1) = g(0) + k$ for some $k\in \Z$ and $g^{(j)}(0) = g^{(j)}(1)$ for all $j$. Then $g$ satisfies $g(x) = kx + c$ for some constant $c\in \R$.
\end{lemma}
\begin{proof}
  Define $h(x) := g(x-1) + k$. Note first that $h$ is still an analytic function, and furthermore that $h(1) = g(0) + k = g(1)$ and $h^{(j)}(1) = g^{(j)}(0) = g^{(j)}(1)$. Hence $g$ and $h$ agree on all derivatives and hence must be equal. As a result $g(x) = g(x-1) + k$, and hence also $g^{(j)}(x) = g^{(j)}(x-1)$.

  Now let's use this equality and evaluate $g$ and its derivatives at the base point $b$ of the series as well as in $b-1$.
  We of course have $g(b) = a_0$. But also $g(b) = g(b-1) + k = a_0 + k + \sum_{j=1}^\infty a_j (-1)^j$. Hence, $\sum_{j=1}^\infty a_j (-1)^j = -k$.
  We do the same for the first derivative: $a_1 = g^{(1)}(b) = g^{(1)}(b-1) = \sum_{j=1}^\infty j a_j (-1)^{j-1}$. After some rewriting we get $\sum_{j=2}^\infty j a_j (-1)^j = 0$ (note the sum starting at $j=2$, and replacing $(-1)^{j-1}$ by $(-1)^j$ by factoring out a global $-1$ factor). In general, doing this with the $l$th derivative, we get the equation 
  $$\sum_{j=l+1}^\infty j!/(j-l)! a_j (-1)^j \ = \ 0. $$
  This then gives us an infinite set of linear equations for the $a_j$ in upper triangular form:
  \begin{equation}
    \begin{tabular}{cllllll}
      $-k$ & $=$ & $-a_1$ & $+a_2$  & $-a_3$ & $+a_4$ & $\ldots$ \\
      $0$  & $=$ &  & $+2a_2$  & $-3a_3$ & $+4a_4$ & $\ldots$ \\
      $0$  & $=$ &  &   & $-6a_3$ & $+12a_4$ & $\ldots$ \\
      & $\vdots$ &&&&&
    \end{tabular}
  \end{equation}
  Using standard induction arguments one can show that this has a unique solution given by $a_1 = k$, and $a_j = 0$ for $j>1$. Hence $g(x) = a_0 + k(x - b) = kx - kb + a_0$. Taking the constant $c=a_0-kb$ finishes the proof.
\end{proof}

Recalling that we defined $g(x) := f(x,x_2,\ldots, x_n)$ we see with the above lemma that $f(x,x_2,\ldots, x_n) = kx + c$, where the $c$ is necessarily $f(0,x_2,\ldots, x_n)$. We can do the exact same argument for the argument $x_2$ as well to arrive at $f(x,x_2,\ldots, x_n) = k_2x_2 + f(x,0,x_3,\ldots,x_n) = kx + k_2 x_2 + f(0,0,x_3,\ldots, x_n)$. Hence, repeating the argument for each parameter we see that $f(x_1,\ldots,x_n) = k_1x_1+\cdots k_n x_n + c$ where $c=f(0,\ldots,0)$. Note that we can write this as $f(\vec x) = \vec k \cdot \vec x + c$ where $\cdot$ is the dot product and $\vec k = (k_1,\ldots, k_n)$.

Now $\Phi = (\Phi_1,\ldots, \Phi_l)$ where $\Phi_j(\vec \alpha) = 2\pi f_j((\vec \alpha \text{ mod } 2\pi)/(2\pi)) \text{ mod } 2\pi = 2\pi (\vec k_j\cdot \frac{\vec \alpha}{2\pi} + \frac{c_j}{2\pi}) \text{ mod } 2\pi = (\vec k_j\cdot \vec\alpha + c_j) \text{ mod } 2\pi$. Writing $\vec c = (c_1,\ldots, c_l)$ and $M = (\vec k_1~\cdots~\vec k_l)$ we see that we then indeed get $\Phi(\vec \alpha) = M\vec \alpha + \vec c \text{ mod } 2\pi$.

\section{Proof of Parametrised completeness}\label{app:completeness}

In this Appendix we will set out to prove Theorem~\ref{thm:completeness-scalars}.

We will need the following facts:
\begin{lemma}
  Let $S$ be a scalar Clifford diagram. Then $S$ is equal to $\sqrt{2}^n e^{i k \frac\pi4}$ for some integers $n,k\in \mathbb{Z}$.
\end{lemma}

\begin{lemma}\label{lem:Clifford-norm}
  Let $D$ be a Clifford diagram. Consider now the diagrams $D_+$ and $D_-$ we get by plugging in a $\ket{+}$, respectively a $\ket{-}$ state into one designated input or output wires of $D$. Then if $D\neq 0$ exactly one of the following is true:
  \begin{itemize}
    \item $\norm{D_+} = \norm{D_-}$.
    \item $\norm{D_+} = 0$.
    \item $\norm{D_-} = 0$.
  \end{itemize}
\end{lemma}
\begin{proof}
  Without loss of generality we can take $D$ to be a state, which is hence (proportional to) a stabiliser state. Then plugging in $\ket{+}$ or $\ket{-}$ into one of the wires corresponds to a post-selection on an $X$ measurement. It is well known that a Pauli measurement on a stabiliser state is either deterministic or unbiased, corresponding to the scenarios described above.
\end{proof}

We will now first study in-depth the behaviour of trivial parameters, and diagrams which just have a single parameter, before moving to the general case.

\begin{definition}
  We say a parametrised diagram $D$ is \emph{singular} when there exists a choice of parameters $\vec\alpha$ such that $D[\vec\alpha] = 0$.
\end{definition}

\begin{lemma}
  Let $D$ be a singular parametrised diagram with just a single parameter. Then that parameter is trivial.
  Furthermore, if there is more than one value of the parameter for which $D$ is zero, then $D=0$. Otherwise the unique value $\alpha$ of the parameter such that $D[\alpha]=0$ is Clifford: $\alpha = k\frac\pi2$.
\end{lemma}
\begin{proof}
  Suppose that for some phase $\alpha$ we have $D[\alpha] = 0$, and let $\beta$ and $\gamma$ be any other phases. Write $\ket{+_\gamma} = a \ket{+_\alpha} + b\ket{+_\beta}$.
  Then $D[\gamma] = a D[\alpha] + bD[\beta] = b D[\beta]$, so that the parameter is indeed trivial.
  Now, if there were a second value of the parameter for which $D$ is zero, then we could have taken $\beta$ to be equal to that value, which would have shown $D[\gamma]=0$. Since $\gamma$ is arbitrary this indeed shows $D=0$.

  So suppose that $\alpha$ is the unique value for which $D[\alpha]=0$. Suppose that $\alpha\neq 0$ and $\alpha\neq \pi$. We will show that then $\alpha = \pm \frac\pi2$. Now with the above argument, taking $\gamma:=\pi$ and $\beta:= 0$ we calculate $D[\pi] = b D[0]$. Since we know that $D[\pi] \neq 0$ and $D[0] \neq 0$, Lemma~\ref{lem:Clifford-norm} shows that necessarily $\lvert b \rvert = 1$. But Lemma~\ref{lem:basis-decomp} also gives us $b = \frac{e^{i\alpha} +1}{e^{i\alpha}-1}$. Hence we must have $\lvert e^{i\alpha}+1\rvert = \lvert e^{i\alpha} - 1\rvert$. Converting the expressions in sine and cosine shows that this only holds when $\alpha=\pm \frac\pi2$.
\end{proof}

\begin{proposition}\label{prop:trivial-parameter-cases}
  Let $D\neq 0$ be a parametrised diagram with a single parameter $\alpha$ and suppose $\alpha$ is trivial. Then exactly one of the following is true:
  \begin{itemize}
    \item $D[0] = 0$ and $D[\alpha] = \frac12(1-e^{i\alpha}) D[\pi]$.
    \item $D[\pi] = 0$ and $D[\alpha] = \frac12(1+e^{i\alpha}) D[0]$.
    \item $D[\frac\pi2] = 0$ and $D[\alpha] = \frac12 (1+e^{i(\alpha+\frac\pi2)}) D[0]$.
    \item $D[-\frac\pi2] = 0$ and $D[\alpha] = \frac12 (1+e^{i(\alpha-\frac\pi2)})D[0]$.
    \item $D[\pi] = D[0]$ and $D[\alpha] = D[0]$.
    \item $D[\pi] = -D[0]$ and $D[\alpha] = e^{i\alpha} D[0]$.
  \end{itemize}
  In particular, if $D$ is non-singular, then $D[\pi] = \pm D[0]$.
\end{proposition}
\begin{proof}
  If $D[0] = 0$, then using $\ket{+_\alpha} = a\ket{+_0} + b \ket{+_\pi}$, we get $D[\alpha] = bD[\pi]$. Using Lemma~\ref{lem:basis-decomp} we see that $b = \frac12(1-e^{i\alpha})$. The other 3 cases where $D$ is singular are similarly calculated. 

  So suppose $D$ is not singular. By Lemma~\ref{lem:Clifford-norm} we must then have $\norm{D[\pi]} = \norm{D[0]}$ and hence $D[\pi] = e^{i\phi} D[0]$. 
  Now because $e^{i\phi}$ must be expressible as a Clifford phase, we must have $\phi = k\frac\pi4$. But actually, since $\phi$ is a relative phase that only gets applied when $\pi$ is inserted in the diagram, and not when $0$ is, we must have $\phi = k\frac\pi2$, since it otherwise would correspond to a T gate. It is straightforwardly verified using the previously used technique, that if $\phi = \pm \frac\pi2$ that then $D[\mp\frac\pi2] = 0$. Hence, as $D$ is non-singular, we must have $\phi = 0$ or $\phi = \pi$, which correspond to the two remaining cases.
\end{proof}

\begin{proposition}
  Let $D_1$ and $D_2$ be parametrised diagrams with a single parameter. Suppose that $D_1[\alpha] = \lambda(\alpha) D_2[\alpha]$ for some scalar function $\lambda:\R\to\CC$ which is non-zero for all $\alpha$. If $\alpha$ is trivial in $D_1$, then $\alpha$ is trivial in $D_2$ and there exists some scalar Clifford diagram $S$ such that either $D'_1 = S\otimes D'_2$ or $D'_1 = S\otimes (D'_2\circ X(\pi))$. In either case we can rewrite, up to a scalar, $D_1[\alpha]$ into $D_2[\alpha]$, using the Clifford rewrite rules.
\end{proposition}
\begin{proof}
  We make some case distinctions. First, if $D_1[\alpha] = 0$ for all $\alpha$, then the same is true for $D_2[\alpha]$ and we are done. So assume $D_1 \neq 0$. By Proposition~\ref{prop:trivial-parameter-cases} there are then 6 cases to check. The first 4 are very similar, so we only do the first of those. Suppose $D_1[0] = 0$. Then because $\lambda(0)\neq 0$ by assumption, we must also have $D_2[0] = 0$.
  Hence, by the first case in Proposition~\ref{prop:trivial-parameter-cases} we have 
  $$\lambda(\alpha) D_2[\alpha] = D_1[\alpha] = \frac12(1-e^{i\alpha})D_1[\pi] = \lambda(\pi) \frac12(1-e^{i\alpha}) D_2[\pi] = \lambda(\pi) D_2[\alpha].$$
  Hence, $\lambda(\alpha) = \lambda(\pi)$ for all $\alpha\neq 0$. Since $D_2[0] = 0$ we might as well take $\lambda(0) = \lambda(\pi)$. Let $S$ denote the scalar Clifford diagram denoting $\lambda(\pi)$. Then $D_1 = S\otimes D_2$.

  The other cases where one of the diagrams is zero for one of the values of the parameter are handled similarly. Hence, assume that $D_1[\alpha] \neq 0$ for all $\alpha$. Then also $D_2[\alpha] \neq 0$. We must then have $D_1[\pi] = \pm D_1[0]$, and $D_2[\pi] = \pm D_2[0]$. Again, we make some case distinctions. 

  Suppose $D_1[\pi] = D_1[0]$ and $D_2[\pi] = D_2[0]$. Then $D_1[\alpha] = D_1[0]$ and $D_2[\alpha] = D_2[0]$, so that $D_1[\alpha] = D_1[0] = \lambda(0) D_2[0] = \lambda(0) D_2[\alpha]$, and we are done.
  If instead $D_1[\pi] = - D_1[0]$ and $D_2[\pi] = - D_2[0]$, then $D_1[\alpha] = e^{i\alpha} D_1[0]$ and $D_2[\alpha] = e^{i\alpha} D_2[0]$, so that we can again verify that $D_1[\alpha] = \lambda(0) D_2[\alpha]$ and we are done.

  Hence, assume that $D_1[\pi] = D_1[0]$, while $D_2[\pi] = - D_2[0]$ (the other remaining case is handled analogously).
  Using that $D_2[\alpha] = e^{i\alpha} D_2[0]$ we then calculate 
  \[\lambda(\alpha)e^{i\alpha} D_2[0] = \lambda(\alpha) D_2[\alpha] = D_1[\alpha] = D_1[0] = \lambda(0) D_2[0],\]
  so that $\lambda(\alpha) = e^{-i\alpha} \lambda(0)$.
  Letting $\overline{D}_2' = D'_2 \circ X(\pi)$ we note that $\overline{D}_2[\alpha] = e^{i\alpha}D_2[-\alpha] = D_2[0]$.
  Hence, $D_1[\alpha] = e^{-i\alpha}\lambda(0) D_2[\alpha] = \lambda(0) D_2[0] = \lambda(0) \overline{D}_2[\alpha]$.
  We see then that ZX proves that $D_1[\alpha]$ and $S\otimes \overline{D}_2[\alpha]$ are equal, for $S = \lambda(0)$.
\end{proof}

\begin{lemma}\label{lem:remove-phase-function}
  Let $D_1[\alpha]$ and $D_2[\alpha]$ be (not necessarily Clifford) parametrised diagrams with just one non-trivial parameter $\alpha$. Suppose there is some scalar function $\lambda:\R\to \CC$ such that $D_1[\alpha] = \lambda(\alpha) D_2[\alpha]$, then $\lambda(\alpha) = \lambda'$ is constant and we have $D_1' = \lambda' D_2'$.
\end{lemma}
\begin{proof}
  Pick some $\alpha$ and $\beta$ distinct from one another. By non-triviality of the parameter we know that $D_1[\alpha]$ is not proportional to $D_1[\beta]$ for this choice of $\alpha$ and $\beta$. This then also implies that $D_2[\alpha]$ is not proportional to $D_2[\beta]$. 
  Let $\gamma$ be any other phase, i.e.~$\gamma\neq \alpha$ and $\gamma\neq \beta$.

  Let $a$ and $b$ be the constants such that $\ket{+_\gamma} = a\ket{+_\alpha} + b \ket{+_\beta}$. These $a$ and $b$ exist because $\ket{+_\alpha}$ and $\ket{+_\beta}$ form a basis. Note that necessarily $a\neq 0$ and $b\neq 0$, as $\gamma\neq \alpha,\beta$.
  Then:
  \[\input{phase-cancel-lemma-pf1.tikz}\]
  But on the other hand:
  \[\input{phase-cancel-lemma-pf2.tikz}\]
  Hence, we must have:
  \[a \lambda(\alpha) D_2[\alpha] + b \lambda(\beta) D_2[\beta] \ =\ a \lambda(\gamma) D_2[\alpha] + b \lambda(\gamma) D_2[\beta]. \]
  Now bring matching terms to each side:
  \[a(\lambda(\alpha) - \lambda(\gamma)) D_2[\alpha] = - b(\lambda(\beta) - \lambda(\gamma)) D_2[\beta].\]
  Since by assumption $D_2[\alpha]$ and $D_2[\beta]$ are not proportional to each other, we must have that each side is equal to $0$. As we furthermore know that $a,b\neq 0$, we must then have $\lambda(\alpha) - \lambda(\gamma) = 0 = \lambda(\beta) - \lambda(\gamma)$. Hence: $\lambda(\alpha) = \lambda(\gamma) = \lambda(\beta)$.
  Since $\gamma$ was an arbitrary phase not equal to $\alpha$ or $\beta$ we see hence that $\lambda(\gamma) = \lambda'$ is a constant. But then $D_1[\alpha] = \lambda' D_2[\alpha]$ for all $\alpha$, and hence the diagrams must be equal when we leave the wire open: $D_1' = \lambda' D_2'$.
\end{proof}
As we also noted in Remark~\ref{rem:finite-values}, to use Proposition~\ref{prop:parameter-equality} we only really need to know that the two parametrised maps agree on two values. The proof of Lemma~\ref{lem:remove-phase-function} furthermore required a third value in order to get the phases to be equal.
We hence note the following corollary.
\begin{corollary}\label{cor:phase-finite-values}
  Let $D_1[\alpha]$ and $D_2[\alpha]$ be parametrised diagrams with just one parameter $\alpha$ that is furthermore non-trivial, and suppose that $D_1[k \frac\pi2] = \lambda(k) D_2[k\frac\pi2]$ for some scalar function $\lambda:\mathbb{Z}\to \CC$ and all $k\in \mathbb{Z}$. Then $D'_1 = \lambda' D'_2$ for some constant scalar $\lambda'$.
\end{corollary}

\begin{proposition*}[Restatement of Prop.~\ref{prop:scalar-is-constant}]
  Let $D_1$ and $D_2$ be (not necessarily Clifford) parametrised diagrams with the same number of parameters and where all the parameters of $D_1$ are non-trivial. Then if $D_1[\vec \alpha] = \lambda(\vec \alpha) D_2[\vec \alpha]$ for all $\vec \alpha$, then $D_1' = \lambda' D_2'$ for some constant scalar $\lambda'$.
\end{proposition*}
\begin{proof}
  For convenience we will write $D_1[\alpha,\vec \alpha]$ for the set of parameters in $D_1$. 
  For an integer string $\vec k\in \{0,1\}^n$ let $\vec\alpha^{\vec k} := (k_1\frac\pi2,\ldots, k_n\frac\pi2)$. We then write $D_i^{\vec k}:=D_i[-, \vec\alpha^{\vec k}]$ for the parametrised diagrams $D_1^{\vec k}$ and $D_2^{\vec k}$ which each have a single non-trivial parameter $\alpha$. Note that $D_1^{\vec k}[\alpha] = \lambda(\alpha,\vec \alpha^{\vec k}) D_2^{\vec k}[\alpha]$ for every $\alpha$. Hence, by setting $\lambda^{\vec k}(\alpha) := \lambda(\alpha, \vec \alpha^{\vec k})$ we satisfy the conditions of Lemma~\ref{lem:remove-phase-function} so that we must have
  $(D_1^{\vec k})' = \lambda_{\vec k} (D_2^{\vec k})'$, where $\lambda_{\vec k}$ is a scalar.

  Now consider the diagrams $D_i[-,\vec \alpha]$ for $i=1,2$, which are like $D_i[\alpha,\vec \alpha]$, but with the wire where $\alpha$ is plugged in left open: 
  \ctikzfig{phased-parameter-equality-pf1}
  Consider now a shorter integer string $\vec k' \in \{0,1\}^{n-1}$ and the one-parameter diagrams $E_i^{\vec k'}[\beta] := D_i[-,\beta, \vec\alpha^{\vec k'}]$. 
  Then $E_1^{\vec k'}[k\frac\pi2] = \lambda_{k,\vec k'} E_2^{\vec k'}[k\frac\pi2]$ so that this satisfies the assumptions of Corollary~\ref{cor:phase-finite-values} (with $\lambda(k):= \lambda_{k,\vec k'}$). We hence note that $\lambda_{k,\vec k'}$ must be independent of $k$, so let's call the scalar just $\lambda_{\vec k'}$. Additionally, we must then have $(E_1^{\vec k'})' = \lambda_{\vec k'} (E_2^{\vec k'})'$.

  We can now repeat the same argument until we have gone through all the parameters. We then conclude that $D'_1 = \lambda' D'_2$ for some constant scalar $\lambda'$.
\end{proof}

\begin{theorem*}[Restatement of Theorem~\ref{thm:completeness-scalars}]
  Let $D_1$ and $D_2$ be two parametrised diagrams with the same number of parameters and where all the parameters of $D_1$ are non-trivial. Then if $D_1[\vec \alpha] = \lambda(\vec \alpha) D_2[\vec \alpha]$ for all $\vec \alpha$, then $D_1' = C\otimes D_2'$ for some Clifford scalar $C$ and we can uniformly rewrite $D_1[\vec \alpha]$ into $C\otimes D_2[\vec \alpha]$ using Clifford rewrites.
\end{theorem*}
\begin{proof}
  By Proposition~\ref{prop:scalar-is-constant} we may assume that $\lambda$ is in fact a constant $\lambda'$.
  Now if we put $\vec \alpha = (0\cdots 0)$ then both $D_1[0\cdots 0]$ and $D_2[0\cdots 0]$ are entirely Clifford diagrams. But as we also have $D_1[0\cdots 0] = \lambda' D_2[0\cdots 0]$ we see that $\lambda'$ must be a scalar that is representable as a Clifford diagram. Denote the Clifford diagram representing $\lambda'$ by $C$. Then we see that $D_1[\alpha] = C\otimes D_2[\alpha]$ for all $\alpha$. Hence by Proposition~\ref{prop:parameter-equality} we see that $D_1' = C\otimes D_2'$ and that these can be rewritten into one another by the Clifford rewrite rules.
\end{proof}

\begin{remark}
  This proof currently requires a strong form of non-triviality, which is that for any choice of $\vec \alpha$, $D[\alpha, \vec \alpha]$ is not proportional to $D[\alpha', \vec \alpha]$, for any choice of $\alpha\neq \alpha'$. Or in the contra-positive form: we already call a parameter trivial when there is at least one choice of $\vec \alpha$ such that there exist $\alpha\neq\alpha'$ with $D[\alpha, \vec \alpha]$ proportional to $D[\alpha',\vec \alpha]$.
  We conjecture that we can generalise the proof to only requiring that $D[\alpha,\vec \alpha] \neq 0$ for any choice of parameter.
\end{remark}

\section{From non-parsimonious reductions to parsimonious ones}\label{app:non-parsimonious}

This section works towards a proof of Proposition~\ref{prop:to-non-parsimonious}.

First, we derive the following stronger consequences from the equation we found in Lemma~\ref{lem:no-cloning-gadget}. Here we say a bit string $\vec x\in \{0,1\}^k$ is \emph{even} when its Hamming weight is even.
\begin{lemma}\label{lem:weight-space}
  Let $z_1,\ldots, z_k \in \Z$ be integers and $c_1,\ldots, c_k\in \R$ some constants, and suppose that $D_1[\alpha] = D_2[z_1\alpha + c_1,\ldots, z_k\alpha + c_k]$. 
  Let furthermore $\ket{A_w} := \sum_{\vec x\,;\, \vec z\cdot \vec x = w} \ket{\vec x}$ be a superposition of all the basis states $\ket{\vec x}$ with equal \emph{weight} $\vec z\cdot \vec x$.
  Then $D_2'\ket{A_w} = 0$ if $w\not\in\{0,1\}$.
  % Then one of the following is true:
  % \begin{itemize}
  %   \item Exactly one of the $z_j$ is odd, and $D_1[\alpha] = D_2[c_1,\ldots,c_{j-1},\alpha+c_j,c_{j+1},\ldots, c_k]$.
  %   \item All the $z_j$ are even and $D_1[\alpha] = D_1[0]$, so that $\alpha$ is trivial.
  %   \item More than one of the $z_j$ is odd and $D_1[\alpha] = e^{i\alpha} D_1[0]$, so that $\alpha$ is trivial.
  % \end{itemize}
\end{lemma}
\begin{proof}
  Start with the equation from Lemma~\ref{lem:no-cloning-gadget}, and without loss of generality assume all the $c_i$ are zero.
  We see that only the right-hand side contains a $\beta$ and that this equation must hold for for all choices of $\beta$. Now, let's plug in $\ket{0}$:
  \begin{equation}\label{eq:no-cloning-pf-zero}
  \input{no-cloning-pf-gadget-zero.tikz}
  \end{equation}
  The state that depends on $\beta$ that is plugged into $D'_2$ on the right-hand side is a particular superposition of computational basis states. First, lets write $\vec z = \vec a;\vec b$ where $\vec a = z_1\cdots z_l$ are the odd $z_j$'s and $\vec b = z_{l+1}\cdots z_k$ are the even ones. 
  First, for the qubits, $l+1,\ldots k$, we see that the state plugged into each can be written as $\ket{0} + e^{i\beta z_j}\ket{1} = \sum_{x_j} e^{i\beta z_j x_j}\ket{x_j}$. Hence, we can write the full state on these qubits $l+1,\ldots, k$ as $\sum_{\vec x} e^{i \beta \vec b\cdot \vec x}\ket{\vec x}$. 
  The first $l$ qubits have a slightly more complicated structure. The X-spider is (up to a scalar factor we will ignore) a superposition over all the even basis states: $\sum_{\vec x \text{ even}} \ket{\vec x}$. Since all the $z_1,\ldots, z_l$ are odd, $x_1+\cdots+ x_l$ being even is the same as $z_1x_1 + \cdots + z_lx_l$ being even. Hence we can write the state on the first $l$ qubits as $\sum_{\vec a\cdot \vec x \text{ even}}e^{i\beta \vec a\cdot \vec x}\ket{\vec x}$. To make the expression for the last $k-l$ qubits $\sum_{\vec x} e^{i \beta\vec b\cdot \vec x}\ket{\vec x}$ match this more, note that we can also write this as $\sum_{\vec x; \vec b\cdot \vec x \text{ even}} e^{i \beta\vec b\cdot \vec x}\ket{\vec x}$ because all the $b_j = z_{l+j}$ are even.
  We can then combine these representations of the first $l$ and last $k-l$ qubits as:
  \begin{equation}\label{eq:weight-superposition}
      \ket{\psi_\beta} \ := \ \sum_{\substack{\vec x\in \B^k \\ \vec z\cdot \vec x \,\equiv\, 0 \text{ mod } 2}} e^{i\beta \vec z \cdot \vec x} \ket{\vec x} \ = \ 
  \sum_{w \text{ even}} e^{i\beta w} \sum_{\substack{\vec x\in \B^k \\ \vec x\cdot \vec z = w}}\ket{\vec x}.
  \end{equation}
  % This is because the X-spider forces the first $x_1,\ldots, x_l$ to have even parity, and since all the $z_1,\ldots, z_l$ are odd, this is the same as $z_1x_1 + \cdots + z_lx_l$ being even. 
  % This works for the last $k-l$ qubits, since $z_{l+1},\ldots, z_k$ are all even, so that we can just add the $z_jx_j$ terms to this expression for $j>l$ without changing the parity. 
  Hence, we see that for $\vec x=x_1\cdots x_k$ we have $\ket{\vec x}$ in the superposition when $\vec z\cdot \vec x$ even. We will call this the \emph{weight} and we set $w=\vec z\cdot \vec x$. In the rightmost expression in Eq.~\eqref{eq:weight-superposition} we grouped together all the terms with the same weight. Note that because some of the $z_j$ can be negative, that the weight can also be negative. We can hence write $\ket{\psi_\beta} = \sum_{w \text{ even}} e^{i\beta w}\ket{A_w}$ where $\ket{A_w} := \sum_{\vec x\,;\, \vec z\cdot \vec x = w} \ket{\vec x}$. 
  Since $D'_2$ maps every $\ket{\psi_\beta}$ to the same thing (namely, to the left-hand side of Eq.~\eqref{eq:no-cloning-pf-zero}, which is $D_1'\ket{0}$), $D'_2$ must map $\ket{\psi_\beta} - \ket{\psi_{\beta'}}$ to zero, for any $\beta,\beta'\in\R$. Let $K$ be the kernel of $D'_2$. Then $\ket{\psi_\beta} - \ket{\psi_{\beta'}} \in K$. Our goal is to show that $\ket{A_w}\in K$ if $k\neq 0,1$. We will do this by looking at different combinations of $\ket{\psi_\beta} - \ket{\psi_{\beta'}}$. Note that:
  \begin{align*}
    \ket{\psi_\beta} - \ket{\psi_{\beta'}} \ &= \ \sum_{\substack{w\neq 0\\w\text{ even}}} (e^{iw \beta} - e^{iw\beta'})\ket{A_w} \ 
    \\&= \ \sum_{\substack{w > 0\\w\text{ even}}}\left[ (e^{iw \beta} - e^{iw\beta'})\ket{A_w} + (e^{-iw \beta} - e^{-iw\beta'}) \ket{A_{-w}}\right].
  \end{align*}
  In particular, if we set $\beta' = -\beta$, we get
  \[
    \sum_{\substack{w > 0\\w\text{ even}}} 2i\sin{w\beta} (\ket{A_w} - \ket{A_{-w}}) \in K
    % \text{and} \quad \sum_{\substack{w > 0\\w\text{ even}}} 2\cos{w\beta} (\ket{A_w} + \ket{A_{-w}}).
  \]
  We show later in Lemma~\ref{lem:independent-vectors} that this is only possible for all $\beta$, if $(\ket{A_w} - \ket{A_{-w}})\in K$ for all even $w\neq 0$. Intuitively this is because the angle doubling formulas for sine are not linear (but are instead polynomial).
  Since then $(\ket{A_w} - \ket{A_{-w}})\in K$ we can add this term multiplied by $e^{-iw \beta} - e^{-iw\beta'}$ to the expression for $\ket{\psi_\beta} - \ket{\psi_{\beta'}}$ above to cancel out the $\ket{A_{-w}}$ term and still get something that is in $K$. This expression then becomes:
  $$
    \sum_{\substack{w > 0\\w\text{ even}}}(e^{iw \beta} + e^{-iw \beta} - e^{iw\beta'} - e^{-iw\beta'})\ket{A_w} \ = \ 
    2\sum_{\substack{w > 0\\w\text{ even}}}(\cos{w\beta} - \cos{w\beta'})\ket{A_w}.
  $$
  Again, by referring to Lemma~\ref{lem:independent-vectors}, this can only be in $K$ for all values of $\beta$ and $\beta'$ when $\ket{A_w}\in K$ itself, for even $w>0$. Combining this with $(\ket{A_w} - \ket{A_{-w}})\in K$, we see then that also $\ket{A_{-w}}\in K$.
  Hence, $\ket{A_w}\in K$ if $w\neq 0$ for even $w$.

  If in Eq.~\eqref{eq:no-cloning-pf-zero} we plugged in $\ket{1}$ instead of $\ket{0}$, we would have gotten an input state that is a combination of the \emph{odd} weight bit string superpositions. To be precise, the input state for a given $\beta$ would be
  \[
    \ket{\psi_\beta} \ = \ e^{-i\beta} \sum_{w \text{ odd}} e^{iw\beta} \ket{A_w} 
    \ = \ \sum_{w \text{ odd}} e^{i(w-1)\beta} \ket{A_w}.
  \]
  The global $e^{-i\beta}$ factor comes from inputting a $\ket{1}$ which absorbs the first $-\beta$ phase, introducing a $e^{-i\beta}$ (global phases are important here since we want to subtract states from each other). Now, as before $\ket{\psi_\beta} - \ket{\psi_{\beta'}}$ is in the kernel $K$ of $D'_2$. We can analyse this similarly to find that when $w\neq 1$ is odd we must have $\ket{A_w}\in K$ . For $w=1$, the term $e^{i(w-1)\beta} = 1$ is constant and so does not appear in the expression $\ket{\psi_\beta} - \ket{\psi_{\beta'}}$, which is why we can't claim anything about it.

  We see then that $D'_2\ket{A_w} = 0$ if $w\not\in \{0,1\}$.
\end{proof}

The following lemma was used in the proof above.
\begin{lemma}\label{lem:independent-vectors}
  Let $S\subseteq V$ be a subspace of a complex vector space $V$, and let $\ket{v_1},\ldots \:\ket{v_m}$ $\in V$ be linearly independent vectors. 
  Then in each of the following situations we must have $\ket{v_k}\in S$ for all $k$:
  \begin{enumerate}
    \item If for all $\beta\in \R$ we have $\sum_{k=1}^m \sin{2k \beta} \ket{v_k} \in S$.
    \item If for all $\beta,\beta'\in \R$ we have $\sum_{k=1}^m (\cos{2k \beta}-\cos{2k\beta'}) \ket{v_k} \in S$.
  \end{enumerate}
\end{lemma}
\begin{proof}
    We only prove the second claim as it is the more complicated one, and the first one follows similarly.
    Recall that the \emph{Chebyshev polynomials} $T_n(x)$ are degree $n$ polynomials satisfying $\cos{n \beta} = T_n(\cos \beta)$. In particular $T_{2k}(x) = a_{2k,2k}x^{2k} + a_{2k,2k-2}x^{2k-2} + \cdots + a_{2k,2}x^2 \pm 1$ where all the $a_{2k,2j}$ are integers and non-zero. Hence, we can rewrite $\cos{2k \beta}$ in terms of powers of $\cos{\beta}$ and group terms accordingly. We hence get
    \[
    \sum_{k=1}^m (\cos{2k \beta}-\cos{2k\beta'})\ket{v_k} = \sum_{k=0}^m (\cos^{2k}{\beta} - \cos^{2k}{\beta'}) \ket{w_k},
    \]
    where $\ket{w_k}$ for $k>0$ is a non-zero integer combination of precisely all $\ket{v_k},\ket{v_{k+1}},\ldots\ket{v_m}$ and $\ket{w_0}$ is a (different) combination of all $\ket{v_1},\ldots, \ket{v_m}$. Now, set $\beta'=\frac\pi2$, so that $\cos{\beta'}=0$. Then all the powers of $\cos{\beta'}$ disappear, except for $k=0$, which is the constant term. The above expression then simplifies to $\sum_{k=1}^m \cos^{2k}{\beta} \ket{w_k}$. Note that the sum now starts at $k=1$, instead of $0$.

    Note now that if we can show that all $\ket{w_k}\in S$ for $k>0$, that we get all $\ket{v_k}\in S$, since the $\ket{w_k}$ form an `upper triangular' combination of the $\ket{v_k}$: $\ket{w_m}$ only contains an integer combination of $\ket{v_m}$, and hence if $\ket{w_m}\in S$, then $\ket{v_m}\in S$. Then we consider $\ket{w_{m-1}}$ which is an integer combination of $\ket{v_m}$ and $\ket{v_{m-1}}$, so that then $\ket{v_{m-1}}$ must also be in $S$. We repeat this until we get to $\ket{w_1}$ which shows that $\ket{v_1}\in S$.

    So it remains to show that $\ket{w_k} \in S$ for all $k$. Writing $x=\cos{\beta}$, we see that we have $\sum_{k=1}^m x^{2k} \ket{w_k} \in S$ for all $x\in [0,1]$. Let $z_1,\ldots, z_m\in (0,1)$ be distinct numbers. Then the matrix $M_{ij} = z_j^{2i} = (z_j^2)^i$ is a Vandermonde matrix and hence is full-rank, so that by elementary row operations it can be reduced to the identity.
    Consider then the states $\ket{\psi_j} = \sum_{k=1}^m z_j^{2k} \ket{w_k} \in S$, i.e.~where we evaluate on $x=z_j$. Considering these states $\ket{\psi_j}$ as the `rows' and applying the same `elementary row operations' (i.e adding these states together in the same way as we would Gauss-Jordan reduce $M_{ij}$), we see that we are left with the states `on the diagonal' $\ket{w_j}$, which are hence also in $S$ as desired.

    Proving this result for the expression with $\sin$ instead of $\cos$ follows similarly, by using the Chebyshev polynomials of the second kind.
\end{proof}

The property from Lemma~\ref{lem:weight-space} is useful for ruling out certain possibilities for a non-parsimonious reduction. Note for instance that there is a unique maximum weight $w_m$ which we get when we set $x_j=1$ iff $z_j\geq 0$ (note that we may take all $z_j\neq 0$ as $z_j=0$ corresponds to the parameter not appearing in that spot, and it hence not being relevant in Lemma~\ref{lem:weight-space}). Hence $\ket{A_{w_m}} = \ket{\vec x_m}$ is just a single computational basis state corresponding to a bit string $\vec x_m$. We then know that $D_2'\ket{\vec x_m} = 0$ if $w_m\neq 0$ and $w_m\neq 1$. 
Note that $w_m=0$ only happens when all the weights are negative. In that case we will instead consider the \emph{minimum} weight, for which its associated $\ket{A_w}$ will also be computational basis state that is sent to zero by $D_2'$. Similarly, if $w_m=1$, then either there is a negative $z_j$, so that we can also choose the minimum weight instead, or there is just a single non-zero $z_j$ which must be equal to $z_j=1$, so that the reduction is already parsimonious. In any case we may without loss of generality assume that there is a $\ket{\vec x_m}$ such that $D_2'\ket{\vec x_m} = 0$.
If we assume that the constants $\vec c$ in the reduction are zero, then $D_2'$ is a Clifford diagram, and then knowing that a certain computational basis vector is sent to zero gives us valuable information.

To formalise this statement we need to use a particular representation of Clifford states that was originally described in~\cite{VandenNest2010Classical}, and recovered in the ZX-calculus in~\cite{poor2023qupit}: the \emph{affine with phases} form that we also saw in Eq.~\eqref{eq:AP-state-decompose}.
In particular, any stabiliser state can be written, up to normalisation, as 
\begin{equation}\label{eq:AP-form}
  \sum_{\substack{\vec x\\ A\vec x = \vec b}} e^{i\frac\pi2 \phi(\vec x)} \ket{\vec x}
\end{equation}
where $A$ is a matrix over the field with two elements $\{0,1\}$, $\vec b$ is some fixed vector, and $\phi$ is a (integer-valued) function consisting of a quadratic part and a linear part in $\vec x$. That is, each stabiliser state is a uniform superposition over some \emph{affine} subset of all computational basis states $\{ \vec x \,|\, A\vec x = \vec b\}$. If we compose a stabiliser state with some diagonal unitaries, we can absorb the action of those unitaries into $\phi$, leaving the superposition of states intact. Note that the matrix $A$ is not uniquely defined, as we only need the set of solutions $\vec x$ to $A\vec x = \vec b$, and hence we can apply row operations to $A$ and $\vec b$ and preserve the same state. Note that the system of equations $A\vec x = \vec b$ is only inconsistent (i.e.~has no solutions) when we are dealing with the zero state, so we will from now on assume that it is in fact consistent.

This form has a straightforward representation as a ZX-diagram:
\begin{equation}\label{eq:AP-state-unfuse}
    \input{AP-state-unfuse.tikz}
\end{equation}
Here each of the internal X-spiders corresponds to an equation of $A$, with the biadjacency matrix between the X-spiders and Z-spiders being equal to $A$ (hence, each $1$ in $A$ corresponds to a wire in the diagram).
  
We are here dealing with Clifford \emph{maps}, instead of \emph{states}, but this just means we have to update Eq.~\eqref{eq:AP-form} to also work with inputs instead of only having outputs.
All of this means that we can write:
$$D_2' := \sum_{\substack{\vec x,\vec y\\A(\vec x,\vec y) = \vec b}} e^{i\frac\pi2 \phi(\vec x, \vec y)} \ketbra{\vec y}{\vec x}.$$
In the ZX-picture, this just corresponds to also having some Z-spiders with inputs, instead of just outputs as in the diagram~\eqref{eq:AP-state-unfuse}. 
We know that $D_2'\ket{\vec x_m} = 0$, and hence we must have $A(\vec x_m, \vec y) \neq \vec b$ for all $\vec y$.
The fact that some $\vec x_m$ is not part of that subspace implies there must be a certain equation in $A$ that witnesses that. We can choose this equation to only involve the $\vec x$ part, and not the $\vec y$ part, as the following lemma demonstrates.
\begin{lemma}\label{lem:affine-restriction}
  Let $\vec x_m\in\{0,1\}^k$, $A=(A_X~A_Y)$ be a $(k+n)\times r$ Boolean matrix and $\vec b\in\{0,1\}^r$ and suppose the system of equations defined by $(A,\vec b)$ has solutions. Suppose furthermore that for all $\vec y\in\{0,1\}^n$ we have $A(\vec x_m,\vec y) = A_X\vec x_m+A_Y\vec y \neq \vec b$. 
  Then there exists $A'=(A_X'~A_Y')$ and $\vec b'$ with the same solution set $\{(\vec x, \vec y)\,|\, A(\vec x, \vec y) = \vec b\} = \{(\vec x, \vec y)\,|\, A'(\vec x, \vec y) = \vec b'\}$, such that $A'$ contains a row $(\vec a_X~\vec 0)$, which only involves the `X side' of $A'$ that witnesses $(\vec x_m, \vec y)$ not being in the solution set. I.e.~supposing this is the $j$th row, we have $\vec a_X\vec x_m \neq b_j'$.
\end{lemma}
\begin{proof}
    We write $A=(A_X~A_Y)$ for $A_X$ the part of the matrix acting on $\vec x$ and $A_Y$ the part that acts on $\vec y$. Then let $A' = (A_X'~A_Y')$ be the matrix we get by applying row operations to $A$ in order to bring $A_Y$ into reduced row echelon form (i.e.~we apply Gaussian elimination to it), and let $\vec b'$ be the transformed version of $\vec b$ (with the same row operations applied to this column vector). We claim now there is a row of $A'$ that witnesses the failure of $(\vec x_m, \vec y)$ to be in the solution set $A'(\vec x,\vec y)=\vec b'$, that has zero components along the $Y$ part of $A'$. 

    Suppose this is not the case. Hence, on all the rows where $A_Y'$ is zero, the $X$ part satisfies the equation when we fill in $\vec x_m$. We will then find a $\vec y$ such that $(\vec x_m, \vec y)$ is in the solution set, which is a contradiction on the starting assumption, so that we must have indeed had a row that witnesses $\vec x_m$ not being part of the solution set where the $Y$ part of $A'$ is zero.
    Since $A_Y'$ is in reduced row echelon form, this means that each non-zero row must contain a pivot column, a $1$ on that column, that is zero on all the other rows. Suppose the $j$th column is a pivot of the $i$th row, then exactly one of the choices $y_j=1$ or $y_j=0$ in $\vec y$ will result in a satisfied equation when we plug in $(\vec x_m, \vec y)$, with the other option giving an unsatisfied equation. Furthermore, choosing $y_j=1$ or $y_j=0$ does not change the satisfiability of any other equation, since $j$ is a pivot. Hence we can set $y_j$ so that the $i$th equation is satisfied when we plug in $(\vec x_m, \vec y)$, without affecting whether any other equation is satisfied. So, by deciding the value of all the pivots in $\vec y$ and setting all the other components of $\vec y$ to zero, we can find a $\vec y$ such that $A'(\vec x_m,\vec y) = \vec b'$, which is the contradiction we seek.
\end{proof}
Hence, we may now assume that we have a representation of our map $D_2'$ where there is an equation in $A$ that only involves $\vec x$, that witnesses that $\vec x_m$ is not part of the superposition (i.e.~is sent to zero).

Let's now work through an example to make clear what this allows us to do.
Suppose we had $D_1[\alpha] = D_2[z_1\alpha, z_2\alpha]$ where $z_1>0$ and $z_2<0$ and where both are odd, so that $\vec x = (1,0)$ is the unique state giving maximum weight (and $\vec x' = (0,1)$ is the unique state giving the minimum weight which we should use instead if $z_1=1$). Hence, assuming $z_1>1$, we have $D'_2\ket{10} = 0$. Then there must be an equation defining the affine subspace of $D_2$ involving the two variables $x_1$ and $x_2$ for which $(1,0)$ is not a solution. There are three possible equations that do this: $x_1 = 0$, $x_2 = 1$ and $x_1\oplus x_2 = 0$. We can split up the cases in terms of how many variables are involved in the equation. First, let's suppose $A$ contains the equation $x_2=1$. In that case the qubit wire leading to the corresponding parameter can be disconnected from the rest of the diagram:
\begin{equation}
  \input{two-parameter-disconnect.tikz}
\end{equation}
We can then use this to simplify the expression of Lemma~\ref{lem:no-cloning-gadget}, as we now know that $I\otimes \ketbra{1}{1}$ stabilises $D_2'$:
\begin{equation}\label{eq:two-parameter-disconnect-2}
  \input{two-parameter-disconnect-2.tikz}
\end{equation}
By Lemma~\ref{lem:no-cloning-gadget}, this is equal to $D_1'$, and hence we see that $D_1[\alpha] = D_2[-\alpha,0]$, which is parsimonious. If instead of the equation $x_2=1$, we had $x_1=0$, we could do an analogous argument.

Now suppose that instead we had the final possibility, $x_1\oplus x_2 = 0$, for which $(1,0)$ is also not a solution. In that case we can also simplify the expression coming from Lemma~\ref{lem:no-cloning-gadget}, as we know that $D_2'$ must be stabilised by the projector onto $\text{span}(\ket{00},\ket{11})$, which is a Z-spider:
\begin{equation}\label{eq:two-parameter-disconnect-3}
  \input{two-parameter-disconnect-3.tikz}
\end{equation}
So in this case we have $D_1[\alpha] = D_1[0]$, so that the parameter is trivial.

This in fact covers all the possibilities for a parameter that appears twice (as there must be an equation involving just two variables present in the representation of $D_2'$, and so the equation either contains just 1 variable, or both variables, which are covered by the above two examples), so we see that either we can reduce it to a parsimonious case, or the parameter is necessarily trivial. Note that we didn't need the fact that the constants $\vec c$ were Clifford, as we only used information about the affine subspace, and not the phases, so that this holds for any affine reduction were variables appear at most twice.
If we have more than one parameter, we can leave the other parameters `plugged in', which just introduce additional phases we can ignore, so that we can repeat this argumentation for each parameter independently:
\begin{proposition}
  Let $D_1[\vec\alpha] = D_2[M\vec \alpha +\vec c]$ be an affine reduction where every column of $M$ contains at most two non-zero elements, and suppose that all parameters $\vec \alpha$ are non-trivial. Then there exists a parsimonious reduction $D_1[\vec \alpha] = D_2[M'\vec \alpha + \vec c]$.
\end{proposition}
\begin{proof}
  We prove by induction on the number of parameters in $D_1$. If there is 1 parameter, then we can write $D_1[\alpha] = D_2[z_1\alpha+c_1, z_2\alpha+c_2]$. Depending on whether $z_1$ and $z_2$ are lesser or greater than 0, $D_2'$ either sends $\ket{00}$, $\ket{01}$, $\ket{10}$, or $\ket{11}$ to zero, as a consequence of Lemma~\ref{lem:weight-space}. In all cases, we know by Lemma~\ref{lem:affine-restriction} that $D_2'$ contains an equation that only involves the two input variables which witnesses that this state is sent to zero. By the analysis above, we can then either conclude that there is a parsimonious reduction, or that the parameter is trivial (which is ruled out by the conditions of the proposition).

  Now, for the induction step, by fixing one parameter to a fixed value, we can include it into the constant $\vec c$, reducing it to the case with one fewer parameter. Then induction shows that each of the other parameters can be included in the reduction parsimoniously. Then working with this new reduction where all parameters are used parsimoniously, except for the one we fixed, we then fix a different parameter, to show that we can also change the reduction to use the original parameter parsimoniously.
\end{proof}

We can do a similar argument for when we have three repetitions of a variable, although the method breaks down when we have four repetitions. So let's consider a reduction $D_1[\alpha] = D_2[z_1\alpha + c_1,z_2\alpha + c_2, z_3\alpha + c_3]$.

Due to Lemma~\ref{lem:affine-restriction}, we now still get an equation restricting the affine subspace of $D_2'$. We can again make a case distinction on how many of the three variables this equation contains. If this involves just one variable, the reduction can be restricted to one where the variable only occurs twice using a version of Eq.~\eqref{eq:two-parameter-disconnect-2}, so that the previous case applies and we are done. If the equation involves two variables, then we can use a version of Eq.~\eqref{eq:two-parameter-disconnect-3} to reduce it to the parsimonious case as well:
\begin{equation}
  \input{three-parameter-disconnect-2.tikz}
\end{equation}
In this particular case then $D_2[\alpha,\beta,\gamma] = D_2[\alpha,\beta+\gamma,0]$ for any $\alpha,\beta,\gamma$ (though note that in the general case, we might also get some minus signs). Hence: $D_1[\alpha] = D_2[z_1\alpha+c_1, (z_2+z_3)\alpha + (c_2+c_3),0]$, so that we have now also reduced to the case of two repetitions of the parameter, and the previous proposition applies.

Now if the restricting equation involves all three variables we need to make a case distinction. If none or only one of the $z_j$ is odd, then by Lemma~\ref{lem:no-cloning-gadget} we already know the reduction can be made parsimonious (or the parameter is trivial). If all three $z_j$ are odd, then the three-qubit version of Eq.~\eqref{eq:two-parameter-disconnect-3} applies, and the parameter in $D_1$ is trivial:
\begin{equation}\label{eq:three-parameter-disconnect}
  \input{three-parameter-disconnect.tikz}
\end{equation}
Then use the diagram from Lemma~\ref{lem:no-cloning-gadget} to relate $D_1$ to the diagram in Eq.~\eqref{eq:three-parameter-disconnect} to show the parameter is indeed trivial.

The remaining case is then when two of the $z_j$ are odd. In that case, the equivalent of Eq.~\eqref{eq:two-parameter-disconnect-3} and Eq.~\eqref{eq:three-parameter-disconnect} toggles the connectivity of the X-spider, so that it is equivalent to the X-spider only being connected to just 1 of the spiders:
\begin{equation}\label{eq:three-parameter-disconnect-3}
  \input{three-parameter-disconnect-3.tikz}
\end{equation}
Hence, in this case we again see the parameter becomes parsimonious.
\begin{proposition}
  Let $D_1[\vec\alpha] = D_2[M\vec \alpha +\vec c]$ be an affine reduction where every column of $M$ contains at most \emph{three} non-zero elements, and suppose that all parameters $\vec \alpha$ are non-trivial. Then there exists a parsimonious reduction $D_1[\vec \alpha] = D_2[M'\vec \alpha + \vec c]$.
\end{proposition}

These techniques continue to work when we have more parameter repetitions, but then they are no longer sufficient to rule out all cases (as far as we can tell).
For any number of parameter repetitions, if the restricting equation contains 1 or 2 variables, we reduce to a case with fewer repetitions, so that we can use induction. For four repetitions the only problematic cases are then when there are exactly 2 odd $z_j$, and the restricting equation contains 3 variables such that this does not contain all odd $z_j$, or the restricting equation contains all 4 variables. We conjecture that using some other consequences of Lemma~\ref{lem:weight-space} we should also be able to prove that when we have this situation that we can still construct a parsimonious reduction.

\end{document}